\documentclass[preprint,12pt,authoryear]{elsarticle}

\usepackage{amsmath}
\usepackage{amssymb}
\usepackage{amsfonts}
\usepackage{cancel}
\usepackage{amsthm}
\usepackage{mathtools}
\usepackage{enumerate}
\usepackage{comment}
\allowdisplaybreaks

\usepackage[T1]{fontenc}
\usepackage[mathscr]{eucal}
\usepackage{dsfont}

\usepackage{fullpage}

\usepackage{setspace}

\usepackage{acronym}
\usepackage{latexsym}
\usepackage{longtable}
\usepackage{paralist}
\usepackage{wasysym}
\usepackage{xspace}

\usepackage[dvipsnames,svgnames]{xcolor}
\colorlet{MyBlue}{DodgerBlue!60!Black}
\colorlet{MyGreen}{DarkGreen!85!Black}

\usepackage[font=small,labelfont=bf]{caption}
\usepackage{subfigure}
\usepackage{tikz}
\usetikzlibrary{calc,patterns}

\usepackage{hyperref}
\hypersetup{
colorlinks=true,
linktocpage=true,
pdfstartview=FitH,
breaklinks=true,
pdfpagemode=UseNone,
pageanchor=true,
pdfpagemode=UseOutlines,
plainpages=false,
bookmarksnumbered,
bookmarksopen=false,
bookmarksopenlevel=1,
hypertexnames=true,
pdfhighlight=/O,
urlcolor=MyBlue!60!black,linkcolor=MyBlue!70!black,citecolor=DarkGreen!70!black, % <--- for screen
pdftitle={},
pdfauthor={},
pdfsubject={},
pdfkeywords={},
pdfcreator={pdfLaTeX},
pdfproducer={LaTeX with hyperref}
}

\numberwithin{equation}{section}  %numberwithin goes before cleverefs when using hyperref
\usepackage[sort&compress,capitalize,nameinlink]{cleveref}
\crefname{app}{Appendix}{Appendices}

\crefrangeformat{equation}{\upshape(#3#1#4)\textendash(#5#2#6)}

\usepackage[textwidth=30mm]{todonotes}

\usepackage{soul}
\setstcolor{red}
\sethlcolor{SkyBlue}

\newcommand{\debug}[1]{#1}

\usepackage[linesnumbered,ruled,vlined]{algorithm2e}

\SetCommentSty{mycommfont} % Edit the font + style of comments in algorithm
\SetAlgoSkip{bigskip} % adjust space before algorithm

\theoremstyle{plain}
\newtheorem{theorem}{Theorem}
\newtheorem{corollary}[theorem]{Corollary}
\newtheorem{lemma}[theorem]{Lemma}
\newtheorem{proposition}[theorem]{Proposition}

\theoremstyle{definition}
\newtheorem{definition}[theorem]{Definition}
\theoremstyle{remark}
\newtheorem{remark}[theorem]{Remark}
\numberwithin{theorem}{section}
\DeclarePairedDelimiter{\braces}{\{}{\}}
\DeclarePairedDelimiter{\bracks}{[}{]}
\DeclarePairedDelimiter{\parens}{(}{)}

\DeclarePairedDelimiter{\abs}{\lvert}{\rvert}

\DeclarePairedDelimiter{\ceil}{\lceil}{\rceil}
\DeclarePairedDelimiter{\floor}{\lfloor}{\rfloor}

\DeclarePairedDelimiterX{\braket}[2]{\langle}{\rangle}{#1,#2}

\DeclarePairedDelimiterX{\inner}[2]{\langle}{\rangle}{#1,#2}
\DeclarePairedDelimiterX{\setdef}[2]{\{}{\}}{#1:#2}

\DeclarePairedDelimiterXPP{\probof}[1]{\Prob}{(}{)}{}{%

#1}

\DeclarePairedDelimiterXPP{\exof}[1]{\Expect}{[}{]}{}{%

#1}

\DeclareMathOperator{\expo}{\debug {e}}

\newcommand{\diff}{\ \textup{\debug d}}

\newcommand{\naturals}{\mathbb{\debug N}}
\newcommand{\reals}{\mathbb{\debug R}}

\newcommand{\run}{\debug n}

\newcommand{\Expect}{\mathsf{\debug E}}
\newcommand{\Prob}{\mathsf{\debug P}}

\newcommand{\proba}{\debug p}
\newcommand{\distr}{\debug F}
\newcommand{\stime}{\debug \tau}
\newcommand{\distrtau}{\debug H}
\newcommand{\intdist}{\debug \Phi}
\newcommand{\exptau}{\debug \mu}

\DeclareMathOperator{\Beta}{\mathsf{\debug{Beta}}}
\DeclareMathOperator{\Binomial}{\mathsf{\debug{Binomial}}}

\DeclareMathOperator{\Poisson}{\mathsf{\debug{Poisson}}}

\DeclareMathOperator{\Unif}{\mathsf{\debug{Unif}}}
\DeclareMathOperator{\Var}{\mathsf{\debug{Var}}}

\newcommand{\per}{\debug t}
\newcommand{\peralt}{\debug s}
\newcommand{\Per}{\debug \ell}

\newcommand{\play}{\debug i}

\newcommand{\pA}{\mathrm{\debug{A}}}
\newcommand{\pB}{\mathrm{\debug{B}}}

\newcommand{\actA}{\debug a}
\newcommand{\actAalt}{\actA'}
\newcommand{\actAaltalt}{\actA''}
\newcommand{\actB}{\debug b}
\newcommand{\actBalt}{\actB'}
\newcommand{\actBaltalt}{\actB''}

\newcommand{\nactions}{\debug K}
\newcommand{\nactionsstar}{\nactions^{\ast}}
\newcommand{\nactionsn}{\nactions_{\run}}
\newcommand{\nactionsA}{\nactions^{\pA}}
\newcommand{\nactionsB}{\nactions^{\pB}}
\newcommand{\nactionsAn}{\nactionsn^{\pA}}
\newcommand{\nactionsBn}{\nactionsn^{\pB}}

\newcommand{\actionsA}{[\nactionsA]}
\newcommand{\actionsB}{[\nactionsB]}
\newcommand{\actionsAn}{[\nactionsAn]}
\newcommand{\actionsBn}{[\nactionsBn]}

\newcommand{\Pay}{\debug U}
\newcommand{\Bimatrix}{\boldsymbol{\Pay}}
\newcommand{\PayA}{\Pay^{\pA}}
\newcommand{\PayB}{\Pay^{\pB}}

\newcommand{\eq}[1]{#1^{\ast}}

\newcommand{\potential}{\debug \Psi}
\newcommand{\nequi}{\debug W}

\DeclareMathOperator{\BRD}{\mathsf{\debug{BRD}}}
\DeclareMathOperator{\cyclet}{\mathsf{\debug{cycle}}}
\DeclareMathOperator{\NE}{\mathsf{\debug{NE}}}

\newcommand{\bigoh}{\mathcal{\debug{O}}}

\newcommand{\ind}{\mathds{\debug 1}}

\newcommand{\smalloh}{\debug o}

\newcommand{\ie}{i.e., }
\newcommand{\eg}{e.g., }
\newcommand{\iid}{i.i.d.\ }
\newcommand{\rhs}{r.h.s.\ }
\newcommand{\lhs}{l.h.s.\ }

\DeclareMathOperator*{\argmax}{arg\,max}

\newcommand{\setA}{\debug A}
\newcommand{\ratioBA}{\debug \alpha}
\newcommand{\qu}{\debug q}

\newcommand{\resp}{\debug R}
\newcommand{\cresp}{\debug r}

\newcommand{\best}{\debug A}
\newcommand{\ppath}{\debug \pi}
\newcommand{\paths}{\debug \Pi}
\newcommand{\bestpath}{\debug G}
\newcommand{\brdNE}{\debug C}
\newcommand{\brdR}{\debug E}
\newcommand{\brddiff}{\debug Z}
\newcommand{\Trap}{\debug M}

\newcommand{\eqpay}{\debug S}
\newcommand{\FTs}{\debug J}
\newcommand{\Xonemax}{\debug L}
\newcommand{\badiid}{\debug D}

\newacro{BRD}{best response dynamics}
\newacro{NE}{Nash equilibrium}
\newacroplural{NE}[NE]{Nash equilibria}
\newacro{PNE}{pure Nash equilibrium}
\newacroplural{PNE}[PNE]{pure Nash equilibria}
\newacro{MNE}{mixed Nash equilibrium}
\newacroplural{MNE}[MNE]{mixed Nash equilibria}
\newacro{PFNE}{prior-free Nash equilibrium}
\newacroplural{PFNE}[PFNE]{prior-free Nash equilibria}
\newacro{KKT}{Karush\textendash Kuhn\textendash Tucker}
\newacro{FIP}{finite improvement property}
\newacro{CLT}{central limit theorem}

\journal{arXiv}

\begin{document}

\begin{frontmatter}

\title{Best-Response Dynamics in Two-Person Random Games with Correlated Payoffs}

\author[labelLuiss]{Hlafo Alfie Mimun}

\affiliation[labelLuiss]{organization={Dipartimento di Economia e Finanza, Luiss University},
            addressline={Viale Romania 32}, 
            city={Roma},
            postcode={00197}, 
            state={RM},
            country={Italy}}

\author[labelSapienza]{Matteo Quattropani}

\affiliation[labelSapienza]{organization={Dipartimento di Matematica ``Guido Castelnuovo'', Sapienza Università di Roma},
            addressline={Piazzale Aldo Moro 5}, 
            city={Roma},
            postcode={00185}, 
            state={RM},
            country={Italy}}

\author[labelLuiss]{Marco Scarsini}

\begin{abstract}
We consider finite two-player normal form games with random payoffs.
Player $\pA$'s payoffs are \iid from a uniform distribution.
Given $\proba\in[0,1]$, for any action profile, 
player $\pB$'s payoff coincides with player $\pA$'s payoff with probability $\proba$ and is \iid from the same uniform distribution with probability $1-\proba$.
This model interpolates the model of \iid random payoff used in most of the literature and the model of random potential games. 
First we study the number of \aclp{PNE} in the above class of games.
Then we show that, for any positive $\proba$, asymptotically in the number of available actions, \acl{BRD} reaches a pure Nash equilibrium with high probability.
\end{abstract}


\end{frontmatter}

%
% Section ----------------------------------
%
\section{Introduction}
\label{se:intro}

%
% Subsection ----------------------------------
%
\subsection{The problem}
\label{suse:problem}

Consider the class of two-person normal-form finite games. 
Some properties hold for the entire class, for instance, the mixed extension of each game in the class admits a \acl{NE} \citep{Nas:PNAS1950,Nas:AM1951}.
Some properties hold generically, for instance, generically the number of \aclp{NE} is finite and odd
\citep{Wil:SJAM1971,Har:IJGT1973}.
Some properties do not hold generically and neither does their negation; for instance having a \acl{PNE} or not having a \acl{PNE} is a not a generic property of finite games. 
Still, it may be relevant to know how likely it is for a finite game to admit a pure equilibrium. 
Along a similar line of investigation, how likely is a recursive procedure---such as \acl{BRD}---to reach a \acl{PNE} in finite time?

One way to formalize these questions is to assume that the game is drawn at random according to some probability measure.
It is not clear what a natural probability measure is in this setting; a good part of the literature on the topic has focused on measures that make the payoffs \iid with zero probability of ties.
Few papers have relaxed this assumption.
For instance, \citet{RinSca:GEB2000}  considered payoff vectors that are \iid across different action profiles, but can have some positive or negative dependence within the same action profile.
\citet{AmiColScaZho:MOR2021}  considered \iid payoffs whose distribution may have atoms and, as a consequence, may produce ties.
\citet{DurGau:AGT2016}  studied the class of random potential games, \ie a class of games that admit a potential having \iid entries.

%
% Subsection ----------------------------------
%
\subsection{Our contribution}
\label{suse:contribution}

In this paper we want to study two-person games with random payoffs where the stochastic model for the payoffs parametrically interpolates the case of \iid payoffs with no ties and the case of random potential games.  
In particular, we start with a model where all payoffs are \iid according to a continuous distribution function (without loss of generality, uniform on $[0,1]$) and we consider an \iid set of coin tosses, one for each action profile. 
If the toss gives head, then the original payoff of the second player is made equal to the payoff of the first player; if the toss gives tail, the payoff remains unchanged.
The relevant parameter is the probability $\proba$ of getting heads in the coin toss.
If $\proba=0$, we obtain the classical model of random games with continuous \iid payoffs. 
If $\proba=1$, we get the model of common-interest random games. 
From the viewpoint of \acp{PNE} any potential game is strategically equivalent to a common-interest game.
Therefore, the above class of games parametrically interpolates the case of \iid payoffs with no ties and the case of random potential games.  
When $\proba$ is small, the game is close to a game with \iid payoffs; when $\proba$ is large, the game is close to a potential game. 

For this parametric class of games we first compute the expected number of \acp{PNE} as a function of $\proba$, and then study its asymptotic behavior as the numbers of actions of the two players diverge, possibly at different speeds. 
It is well known \citep{Pow:IJGT1990} that, as the number of action increases, the asymptotic distribution of the number of \acp{PNE} is a $\Poisson(1)$ distribution, for \iid random payoffs.
Our result shows an interesting phase transition around $\proba=0$, in the sense that for every $\proba>0$ the expected number of \acp{PNE} diverges. 

We then consider \ac{BRD} for the above class of games. 
\citet{DurGau:AGT2016} considered \ac{BRD} for random potential games with an arbitrary number of players and the same number of actions for each player. In this class of games a \ac{PNE} is reached by a \ac{BRD} in finite time. 
\citet{DurGau:AGT2016} studied the asymptotic behavior of the expectation of this random time.
In our paper we first consider potential games and we compute the distribution of the time that the \ac{BRD} needs to reach a \ac{PNE}.
Moreover we compute exactly the first two moments of this random time, when the two players have the same action set. 

\citet{AmiColHam:ORL2021} showed that, for games with \iid continuous payoffs, when players have the same action set, as the number of actions increases, the probability that a \ac{BRD} reaches a \ac{PNE} goes to zero. 
Here we generalize the result of \citet{AmiColHam:ORL2021} to the case of possibly different action sets for the two players. 
Moreover, we prove that for every positive $\proba$, asymptotically in the number of actions, a \ac{BRD} reaches a \ac{PNE} in finite time with probability arbitrarily close to $1$. 
Again this shows a phase transition in $\proba=0$ for the behavior of the \ac{BRD}.

%
% Subsection ----------------------------------
%
\subsection{Related literature}
\label{suse:related-literature}

Games with random payoffs have been studied for more than sixty years. 
We refer the reader to \citet{AmiColScaZho:MOR2021,HeiJanMunPanScoTar:IJGT2023} for an extensive survey of the literature on the topic. 
Here we mention just some recent papers and some articles that are more directly connected with the results of our paper. 
\citet{Pow:IJGT1990} proved that in random games with \iid payoffs having a continuous distribution, as the number of actions of at least two players diverges, the asymptotic distribution of the number of \acp{PNE} is $\Poisson(1)$.
\citet{Sta:GEB1995} computed the exact nonasymptotic form of this distribution, from which the result in \citet{Pow:IJGT1990} can be obtained as a corollary.
\citet{RinSca:GEB2000} retained the \iid assumptions for payoff vectors corresponding to different action profiles, but allowed dependence for payoffs within the same profile.    
They proved an interesting phase transition in terms of the payoffs' correlation: asymptotically in either the number of players or the number of actions, for negative dependence the number of \acp{PNE} goes to $0$, for positive dependence it diverges, and for independence it is $\Poisson(1)$, as proved by \citet{Pow:IJGT1990}.
\citet{BalRinSte:PSM1989} studied the distribution of the number of local maxima on a graph, which---by choosing a suitable graph---can be translated into the number of \acp{PNE} in a random potential game.

\citet{PeiTak:GEB2019} studied point-rationalizable strategies in two-person random games. 
Since the number of point-rationalizable strategies for each player is weakly larger than the number of \acp{PNE}, they were interested in the typical magnitude of the difference between these two numbers.
A game is dominance solvable if iterated elimination of strictly dominated strategies leads to a unique action profile, which must be a \ac{PNE}.
\citet{AloRudYar:arXiv2021} used recent combinatoric results to prove that the probability that a two-person random game is dominance solvable vanishes with the number of actions.

Several papers studied the behavior of various learning dynamics \ac{BRD} in games with random payoffs. 
For instance, \citet{GalFar:PNAS2013} studied a type of reinforcement learning called experience-weighted attraction in two-person games and showed the existence of three different regimes in terms of convergence to equilibria.
\citet{SanFarGal:SR2018} extended their analysis to games with an arbitrary finite number of players.
\citet{PanHeiFar:SA2019} compared through simulation the behavior of various adaptive learning procedures in games whose payoffs are drawn at random.
\citet{HeiJanMunPanScoTar:IJGT2023} compared the behavior of \ac{BRD} in games with random payoffs, when the order of acting players is fixed vs when it is random and they showed that, asymptotically in either the number of players or the number of strategies, the fixed-order \ac{BRD} converges with vanishing probability, whereas the random-order does converge to a \ac{PNE} whenever it exists.
Similar results were obtained by \citet{WieHei:DGA2022}.

\citet{CouDurGauTou:NetGCoop2014,DurGau:AGT2016,DurGarGau:PE2019} focused on random potential games and measured the speed of convergence of \ac{BRD} to a \ac{PNE}. 
\citet{AmiColHam:ORL2021} dealt with two-person games where the players have the same action set and   payoffs are \iid with a continuous distribution. 
They compared the behavior of \acl{BRD} and better response dynamics.
They proved that, asymptotically in the number of actions, the first reaches a \ac{PNE} only with vanishing probability, whereas the second does reach it, whenever it exists.
\citet{AmiColScaZho:MOR2021} studied a class of games with $n$ players and two actions for each player where the payoffs are \iid but their distribution may have atoms.  
They proved that the relevant parameter for the analysis of this class of games is the probability of ties in the payoffs, called $\alpha$.
They showed that, whenever this parameter is positive, the number of \acp{PNE} diverges, as $n\to\infty$ and proved a central limit theorem for this random variable. 
Moreover, using percolation techniques, they studied the asymptotic behavior of \ac{BRD}, as a function of $\alpha$, and they showed a phase transition at $\alpha=1/2$.
\citet{JohSavScoTar:arXiv2023} considered the class of games whose random payoffs are \iid with a continuous distribution; they showed that in almost every game in this class that has a \acl{PNE}, asymptotically in the number of players, \acl{BRD} can lead from every action profile that is not a \acl{PNE} to every \acl{PNE}.    

Potential functions in games were introduced by \citet{Ros:IJGT1973} and their properties were extensively studied by \citet{MonSha:GEB1996}.
Among them, existence of \acp{PNE} and convergence to one of these equilibria of the most common learning procedures, including \ac{BRD}.

\citet{GoeMirVet:FOCS2005} introduced the concept of sink equilibrium. 
Sink equilibria are strongly connected stable sets of action profiles that are never abandoned once reached by a \ac{BRD}. 
A sink equilibrium that is not a \ac{PNE} is what in this paper is called a trap.

\citet{FabJagSha:TCS2013} studied the class of weakly acyclical games, \ie the class of games for which from every action profile, there exists some better-response improvement path that leads from that action profile to a \ac{PNE}.
This class includes potential games and dominance solvable games as particular cases.

The goal of our paper is to consider probability measures on spaces of finite noncooperative games that go beyond the usual assumption of \iid payoffs.
In particular, we  define a parametric family of probability measures that interpolates random games with \iid payoffs and random potential games.  
The  interpolation is achieved locally  by acting on each action profile of the game and replacing---with some fixed probability  and independently across profiles---the payoff of the second player with the payoff of the first player.
A different  interpolation could be achieved by considering a convex combination of a game with \iid payoffs and a random potential game. 
This was done, \eg in \citet{RinSca:GEB2000}, 
where in each action profile the payoffs are obtained by summing a Gaussian vector with \iid components and an independent Gaussian vector with identical components (which plays the role of the \emph{random potential}).   
This approach is somehow comparable with the idea of decomposing the space of finite games proposed by \citet{CanMenOzdPar:MOR2011}. 
This decomposition was then used by \citet{CanOzdPar:GEB2013} to analyze \ac{BRD} in games that are close to potential games.

%
% Subsection ----------------------------------
%

\subsection{Organization of the paper}
\label{suse:organization}

\cref{se:preliminaries} introduces some basic game theoretic concepts.
\cref{se:number-PNE} defines the parametric family of distributions on the space of games and deals with the  number of \ac{PNE} in games with random payoffs. 
\cref{se:BRD} studies the behavior of \ac{BRD} in games with random payoffs in the above parametric class, for different values of the parameter.  
\cref{se:proofs} contains all the proofs.
Conclusions and open problems can be found in \cref{se:conclusions}.
\cref{se:symbols} lists the symbols used throughout the paper.
\cref{se:beta} contains two well-known results about the Beta distribution. 

%
% Subsection ----------------------------------
%

\subsection{Notation}
\label{suse:notation}

Given an integer $\run$, the symbol $[\run]$ indicates the set $\braces{1,\dots,\run}$.
Given a finite set $\setA$, the symbol $\abs{\setA}$ denotes its cardinality.
The symbol $\sqcup$ denotes the union of disjoint sets.
We use the notation $x \wedge y \coloneqq \min\braces{x,y}$.
The symbol $\xrightarrow{\Prob}$ denotes convergence in probability.

Given two nonnegative sequences $h_{\run},g_{\run}$, we use the following common asymptotic notations:
\begin{align}
\label{eq:small-oh}
h_{\run}=\smalloh(g_{\run}) \quad&\text{if}\quad \lim_{\run\to\infty} \frac{h_{\run}}{g_{\run}} = 0,\\
\label{eq:big-oh}
h_{\run}=\bigoh(g_{\run}) \quad&\text{if}\quad \limsup_{\run\to\infty} \frac{h_{\run}}{g_{\run}} < \infty,\\
\label{eq:Omega}
h_{\run}=\Omega(g_{\run}) \quad&\text{if}\quad \liminf_{\run\to\infty} \frac{h_{\run}}{g_{\run}} > 0,\\
\label{eq:omega}
h_{\run}=\omega(g_{\run}) \quad&\text{if}\quad \lim_{\run\to\infty} \frac{h_{\run}}{g_{\run}} = \infty,\\
\label{eq:Theta}
h_{\run}=\Theta(g_{\run}) \quad&\text{if}\quad h_{\run}=\bigoh(g_{\run})  \quad\text{and}\quad h_{\run}=\Omega(g_{\run}).
\end{align}

%
% Section ----------------------------------
%

\section{Preliminaries}
\label{se:preliminaries}

We consider two-person normal-form games where, for $\play\in\braces{\pA,\pB}$, player~$\play$'s action set is $[\nactions^{\play}] \coloneqq\braces*{1,\dots,\nactions^{\play}}$ and
$\Pay^{\play} \colon \actionsA \times \actionsB \to \reals$ is player~$\play$'s payoff function.
The game is defined by the payoff bimatrix 
\begin{equation}
\label{eq:bimatrix}
\Bimatrix \coloneqq \parens{\Bimatrix^{\pA},\Bimatrix^{\pB}},  
\end{equation}
where, for $\play\in\braces{\pA,\pB}$,
\begin{equation}
\label{eq:payoff-matrix}
\Bimatrix^{\play} \coloneqq \parens{\Pay^{\play}(\actA,\actB)}_{\actA\in\actionsA,\actB\in\actionsB}.
\end{equation}

A \acfi{PNE}\acused{PNE} of the game is a pair $\parens{\eq\actA,\eq\actB}$ of actions such that, for all $\actA\in\actionsA, \actB\in\actionsB$ we have
\begin{equation}
\label{eq:PNE}
\PayA(\eq\actA,\eq\actB) \ge \PayA(\actA,\eq\actB) 
\quad\text{and}\quad
\PayB(\eq\actA,\eq\actB) \ge \PayB(\eq\actA,\actB). 
\end{equation}

As is well known,  \acp{PNE} are not guaranteed to exist. 
A class of games that admits \acp{PNE} is the class of \emph{potential games}, \ie games for which there exists a \emph{potential function} $\potential \colon \actionsA \times \actionsB \to \reals$ such that for all $\actA,\actAalt\in\actionsA$, for all $\actB,\actBalt\in\actionsB$, we have
\begin{subequations}
\label{eq:potential}
\begin{equation}
\PayA(\actA,\actB) - \PayA(\actAalt,\actB) =
\potential(\actA,\actB) - \potential(\actAalt,\actB), 
\end{equation}
\begin{equation}
\PayB(\actA,\actB) - \PayB(\actA,\actBalt) =
\potential(\actA,\actB) - \potential(\actA,\actBalt).
\end{equation}
\end{subequations}
Games of \emph{common interest}, \ie games for which $\PayA = \PayB$, are a particular case of potential games.
As far as \acp{PNE} are concerned, every potential game is strategically equivalent to a common interest game, for instance to the game where $\PayA = \PayB = \potential$.
For the properties of potential games with an arbitrary number of players, we refer the reader to \citet{MonSha:GEB1996}.

Given a finite game, it is interesting to see whether an equilibrium can be reached iteratively by allowing players to deviate whenever they have an incentive to do so. 
In particular, we will consider a procedure where, starting from a fixed action profile, players in alternation choose their best response to the other player's action.
If the procedure gets stuck in an action profile, then it has reached a \acl{PNE}.
In general, there is no guarantee that this occurs. 

Assume that the payoffs of each player are all different, \ie$ \PayA(\actA,\actB) \neq \PayA(\actAalt,\actB)$ for all $\actA \neq \actAalt$ and all $\actB$ (and similarly for the second player).
The \acfi{BRD}\acused{BRD} is a learning algorithm taking as input a two-player game $(\PayA, \PayB)$ and a starting action profile $\parens{\actA_{0},\actB_{0}}$. For each $\per\ge0 $ we consider the process $\BRD(\per)$ on $\actionsA \times \actionsB$ such that 
\begin{align}
\label{eq:BRD}
\BRD(0) &= \parens{\actA_{0},\actB_{0}} 
\intertext{and, if $\BRD(\per)=\parens{\actAalt,\actBalt}$, then, for $\per$ even,}
\BRD(\per+1) &= \parens{\actAaltalt,\actBalt},
\intertext{where 
$\actAaltalt\in\argmax_{\actA\in\actionsA} \PayA(\actA,\actBalt) \setminus \braces*{\actAalt}$, if the latter set is not empty, otherwise}
\BRD(\per+1) &= \BRD(\per); 
\intertext{for $\per$ odd,}
\BRD(\per+1) &= \parens{\actAalt,\actBaltalt},
\intertext{where 
$\actBaltalt\in\argmax_{\actB\in\actionsB} \PayB(\actAalt,\actB)  \setminus \braces*{\actBalt}$, if the latter set is not empty, otherwise}
\BRD(\per+1) &= \BRD(\per).
\end{align}
It is easy to see that, if, for some 
positive  $\hat\per$, we have
\begin{equation}
\label{eq:BRD-PNE}
\BRD(\hat\per) = \BRD(\hat\per+1) = 
\parens{\eq\actA,\eq\actB},
\end{equation}
then $\BRD(\per)=\parens{\eq\actA,\eq\actB}$ for all $\per\ge\hat\per$ and  $\parens{\eq\actA,\eq\actB}$ is a \ac{PNE} of the game. 

The algorithm stops when it visits an action profile for the second time. 
If this profile is the same as the one visited at the previous time, then a \ac{PNE} has been reached. 

Inspired by the concept of \emph{sink equilibrium} of \citet{GoeMirVet:FOCS2005}, we give a definition of trap in a way that is suitable for our two-player environment.    

\begin{definition}
\label{de:trap}
A \emph{trap}  is a finite set $\Trap$ of action profiles such that 
\begin{enumerate}[(a)]
\item
$\abs{\Trap}\ge 2$, 

\item
if $\BRD(\per) \in \Trap$, then $\BRD(\per+1) \in \Trap$;

\item
for every $(\actA,\actB)\in\Trap$, there exists $\per$ such that $\BRD(\per+k\abs{\Trap})=(\actA,\actB)$, for every $k \in \naturals$.

\end{enumerate} 
\end{definition}

Moreover, the definition of \ac{BRD} implies that for every trap $\Trap$ we have $\abs{\Trap}\ge 4$ and  $\abs{\Trap}$ even.
Moreover, for every game $(\PayA,\PayB)$ and every initial profile $(\actA,\actB)$, the \ac{BRD} eventually visits a \ac{PNE} or a trap in finite time, say $\stime$. 

Even if the game admits \ac{PNE}, there is no guarantee that a \ac{BRD} reaches one of them; it could cycle over a trap, \ie it could start to periodically visit the same set of action profiles and never stabilize.
On the other hand, if the game is a potential game, then a \ac{BRD} always reaches a \ac{PNE}.
This is due to the fact that at every iteration of the \ac{BRD} the payoff of one player increases, and so does the potential. 
Since the game is finite,  in finite time the \ac{BRD} reaches a local maximum, which is a \ac{PNE} \citep[see, \eg][proposition~4.4.6]{KarPer:AMS2017}.

The goal of this paper is to study the number of \acp{PNE} and the behavior of \ac{BRD} in a ``typical'' game. 
To make sense of the above sentence, we need to formalize the meaning of the term typical.
The approach that we will take is stochastic. 
That is, we will assume the bimatrix $\Bimatrix$ to be random and drawn from a distribution that will be specified later.
In any game with random payoffs, the set $\NE$ of \aclp{PNE} is a (possibly empty) random set of action profiles, \ie a random subset of $\actionsA \times \actionsB$. 
Therefore, since the game is finite, the number of \acp{PNE} is an integer-valued  random variable.
Moreover, we will be able to speak about the probability that a \ac{BRD} converges (to a \ac{PNE}).

%
% Section ----------------------------------
%

\section{Number of \aclp{PNE} in random games}
\label{se:number-PNE}
As mentioned in the Introduction,  several attempts have been made in the literature to put a probability measure on a space of games.
Most of the existing papers assume all the entries of $\Bimatrix$ to be \iid with a continuous marginal distribution.
There are some notable exceptions to the independence assumption. 
\citet{RinSca:GEB2000} considered a setting where the payoff vectors of different action profiles have a continuous distribution and are \iid\!\!, but some dependence is allowed within each profile.    
\citet{DurGau:AGT2016} studied  random potential games where the entries of the potential are \iid with a continuous distribution.

Our stochastic model is quite general, since---in a sense that will be made precise---it interpolates the \iid payoffs and the random potential. 

By definition, the concept of \acl{PNE} is ordinal, that is, if all payoffs in a game are transformed according to a strictly increasing function, then the set of \aclp{PNE} remains the same.
Assume that each entry of  $\Bimatrix$ has a marginal distribution that is uniform on the interval $[0,1]$.
Its distribution function will be denoted by $\distr$.
The above consideration implies that this uniformity assumption is without loss of generality, \ie any other continuous distribution would produce the same conclusions.

Start with $\parens{\Bimatrix^{\pA},\Bimatrix^{\pB}}$, where all the entries are \iid with distribution $\distr$.
Then, for each action profile $\parens{\actA,\actB}$, with probability $\proba$ set $\PayB(\actA,\actB)$ to be equal to $\PayA(\actA,\actB)$, independently of the other action profiles.
In other words, for every pair $\parens{\actA,\actB}$,
\begin{itemize}
\item
with probability $1-\proba$,  the random payoffs $\PayA(\actA,\actB)$ and $\PayB(\actA,\actB)$ are independent,

\item
with probability $\proba$, we have $\PayA(\actA,\actB) = \PayB(\actA,\actB)$. 

\end{itemize}
The larger $\proba$, the closer the game is to a potential game. The smaller $\proba$, the closer the game is to a random game with \iid payoffs.
The game whose payoff bimatrix is obtained as above will be denoted by $\Bimatrix(\proba)$.

We now compute the expected number of \acp{PNE} in the above-defined class of random games. 

\begin{proposition}
\label{pr:number-PNE}
If $\nequi$ is the number of \acp{PNE}  in the game $\Bimatrix(\proba)$, then
\begin{equation}
\label{eq:E-W-n}
\Expect\bracks*{\nequi} =  \proba \frac{\nactionsA\nactionsB}{\nactionsA
+\nactionsB-1}+(1-\proba).
\end{equation}
\end{proposition}

The analysis of this class of games is quite complicated for fixed $\nactionsA,\nactionsB$.
Therefore, as it is done in much of the literature, we will take an asymptotic approach, letting the number of actions grow.
More formally, we will consider a sequence $(\Bimatrix_{\run})_{\run
\in\naturals}$ of payoff bimatrices, where the numbers of actions in game $\Bimatrix_{\run}$ are $\nactionsAn$ and $\nactionsBn$, and these two integer sequences are increasing in $\run$ and diverge to $\infty$.
In particular, we allow the number of actions of the two players to diverge at different speeds.

The following proposition shows the asymptotic behavior of the random number of \acp{PNE} where the parameter $\proba$ may vary with $\run$.
We write $\proba_\run$ to highlight this dependence.
In what follows, every asymptotic equality holds for $\run\to\infty$. 
The proof uses a second-moment argument.

For every $\run\in\naturals$, let  $\nequi_{\run}$ be the number of \acp{PNE} in the game $\Bimatrix_{\run}$ and 
\begin{equation}
\label{eq:K-n}
\nactionsn \coloneqq \min(\nactionsAn,\nactionsBn).
\end{equation}

\begin{proposition}
\label{pr:number-PNE-gen}
If $\proba_{\run} = \omega(1/\nactionsn)$, then 
\begin{equation}
\label{eq:W-n-to}
\frac{\nactionsAn+\nactionsBn}{\proba_{\run}\nactionsAn\nactionsBn}\nequi_{\run} \xrightarrow{\Prob} 1.
\end{equation}
\end{proposition}

The following corollary deals with the case of fixed $\proba$.    

\begin{corollary}
\label{co:ratioBA-number}
If $\proba_{\run}=\proba$ for all $\run\in\naturals$  and $\nactionsBn=\omega(\nactionsAn)$, then
\begin{equation}
\label{eq:W-n-omega}
\frac{\nequi_{\run}}{\nactionsAn} \xrightarrow{\Prob} \proba.
\end{equation}

If $\nactionsBn=\ratioBA_{\run}\nactionsAn$, with $\ratioBA_{\run}\to\ratioBA$, then 
\begin{equation}
\label{eq:W-n-equal}
\frac{\nequi_{\run}}{\nactionsAn} \xrightarrow{\Prob} \frac{\ratioBA}{\ratioBA+1}\proba.
\end{equation}

In particular, if $\nactionsAn=\nactionsBn$, then 
\begin{equation}
\label{eq:W-n-ratioBA}
\frac{\nequi_{\run}}{\nactionsn} \xrightarrow{\Prob} \frac{\proba}{2}.
\end{equation}
\end{corollary}

When $\proba=0$, \ie the payoffs are \iid the number of \acp{PNE} converges in distribution to a Poisson with parameter $1$ \citep[see][]{Pow:IJGT1990}.
When the payoffs within the same action profile are positively correlated, \citet{RinSca:GEB2000} showed that the number of \acp{PNE} diverges.
A similar phenomenon happens here when $\proba>0$.

%
% Section ----------------------------------
%

\section{Best Response Dynamics}
\label{se:BRD}
We now want to study the behavior of \ac{BRD} in the class of random games introduced in \cref{se:number-PNE}.
First notice that the continuity of $\distr$ implies that the probability of ties in the payoffs of the same player is zero; as a consequence, once the game is realized, a \ac{BRD} is almost surely deterministic. 
In this respect, the symbols $\Prob$ and $\Expect$ refer solely to the randomness of the payoffs, not to any randomness in the \ac{BRD}.
Moreover, the symmetry of our model implies that, without loss of generality, we can assume  the starting position of the \ac{BRD} to be any fixed profile.
In the rest of the paper, without loss of generality, the starting point of any \ac{BRD} will always be the profile $(1,1)$, \ie for every $\run\in\naturals$ 
\begin{equation}
    \Prob(\BRD_{\run}(0)=(1,1))=1.
\end{equation}

\subsection{\ac{BRD} and related stopping times}\label{suse:BRD-tau}

As stressed above, for a given realization of the payoffs and a starting point, the \ac{BRD} is a deterministic algorithm that decides its next step only on the basis of \emph{local information}. 
In what follows we will exploit this fact by revealing the players' payoffs only when this information is needed to select the next position of the \ac{BRD}. 
This whole process, in which the \ac{BRD} moves on a sequentially sampled random game, can be thought of as a non-Markov stochastic process, and the time at which the \ac{BRD} stops can be seen as a stopping time for such a stochastic process.

We will focus our attention on the distribution of the first time the \ac{BRD} reaches a \ac{PNE}. 
For the sake of brevity, we write $\NE_{\run}$ for the (random) set of \acp{PNE} in the game $\Bimatrix_{\run}(\proba_{\run})$, and we define
\begin{equation}
\label{eq:tau-NE}
\stime_{\run}^{\NE} \coloneqq \min\braces*{\per \colon \BRD_{\run}(\per) \in \NE_{\run}}.
\end{equation}
In words, $\stime_{\run}^{\NE}$ is the first time the process $\BRD_{\run}(\per)$ visits a \ac{PNE}.  
Notice that the first step of the \ac{BRD} is somehow different from the following, because, by the definition of the model, at time $0$ no  player is assumed to be already in a best response. 
Contrarily, for any $\per\ge 1$ odd (resp. even) we have that the first (resp. second) player is in a best response, and the other player's action can be changed at the following step, if it is not itself a best response. 
As a consequence of this fact, the forthcoming definitions that depend on the step $\per$ of the \ac{BRD}, require a special treatment for the first few steps, \ie $\per=0,1,2$. 
Clearly, this issue could be solved by assuming that at $\per=0$ the second player is in best response. 
The latter assumption has been made, \eg in \cite{AmiColHam:ORL2021}. 
For the sake of generality, we prefer to avoid this assumption and rather treat the case of $\per=0,1,2$ separately.
We will make use of the following sequence of random sets:
\begin{align}
\label{eq:response}
\begin{split}
\resp_{\run}(\per) &\coloneqq 
\left\{(\actA,\actB) \colon \text{either }\BRD_{\run}(\peralt)=(\actAalt,\actB) \text{ or }
\BRD_{\run}(\peralt)=(\actA,\actBalt),\right.
\\
&\left.\quad \text{ for some }\actAalt\in\actionsAn, \actBalt\in\actionsBn \text{ and }1\le \peralt\le\per\right\},\quad \per\ge 1,
\end{split} \\
\label{eq:response-0}
\resp_{\run}(0) &\coloneqq \braces*{(\actA,1):\ \actA\in\actionsAn}.
\end{align}
Roughly speaking, the random set $\resp_{\run}(\per)$ is the set of all rows and columns that contain one element that has been visited by the \ac{BRD} up to time $\per$. More precisely, for $\per\ge 1$, in order to determine where $\BRD_\run(\per)$ is and whether it is in a \ac{PNE} or not, for each action profile $(\actA,\actB) \in \resp_{\run}(\per)$, at least one of the payoffs $\PayA(\actA,\actB)$ and $\PayB(\actA,\actB)$, at some time $\peralt\le \per$, needs  to be revealed.

For instance, $\BRD_\run(\per) = (\actA,1)$ for some $\actA\in\actionsAn$. 
To determine whether this profile $(\actA,1)$ is a \ac{PNE} or not, the payoff $\PayB(\actA,1)$ has to be compared with all payoffs $\PayB(\actA,\actB)$ for each $\actB\in\actionsBn$.

To better understand the above definition, consider the random time 
\begin{equation}
\label{eq:def-tau-R}
\stime^{\resp}_\run\coloneqq \min\braces*{\per\ge 2:\ \BRD_\run(\per)\in\resp_\run(\per-2)} .
\end{equation}
By the  definitions in \eqref{eq:response} and \eqref{eq:response-0}, at time $\stime_{\run}^\resp$
\begin{enumerate}[(i)]
    \item
    \label{item:NE} 
    either the \ac{BRD} has reached an equilibrium, \ie $\BRD_\run(\stime^{\resp}_\run)=\BRD_\run(\stime^{\resp}_{\run}-1)\in\NE_\run$;
    
    \item
    \label{item:trap} 
    or the \ac{BRD} has reached a trap, \ie there exists some $\per\le \stime_{\run}^\resp-3$ such that $\BRD_\run(\stime_{\run}^\resp+1)=\BRD_\run(\per)$.
\end{enumerate}
An example of the first steps of the BRD is given in \cref{Fig:tauR}. 
The red (respectively blue) dots are the action profiles whose payoff is compared by the row (respectively column) player. 
The red (respectively blue) lines helps visualize the action profiles that the \ac{BRD} considers when the row (respectively column) player is active in the \ac{BRD}. 
The intersections between these lines are action profiles in which the payoffs of both the column and row player have been examined by the \ac{BRD}. 
Such points are represented by two overlapping dots of different size: the color of the biggest (respectively smallest) one is associated to the first (respectively second) player that compares the payoff of such action profiles. 
The numbers indicate the positions of the \ac{BRD} at the various times. 
The odd (respectively even) numbers are in red (respectively blue) since the associated action profiles are visited for the first time when the row (respectively column) player is active. 
The left side of \cref{Fig:tauR} shows an instance of what occurs in case~\eqref{item:NE} in the above list, whereas the right side of the figure contains a  graphical explanation of what happens in case~ \eqref{item:trap}.

\begin{figure}[h]
    \centering
    \includegraphics[width=14cm]{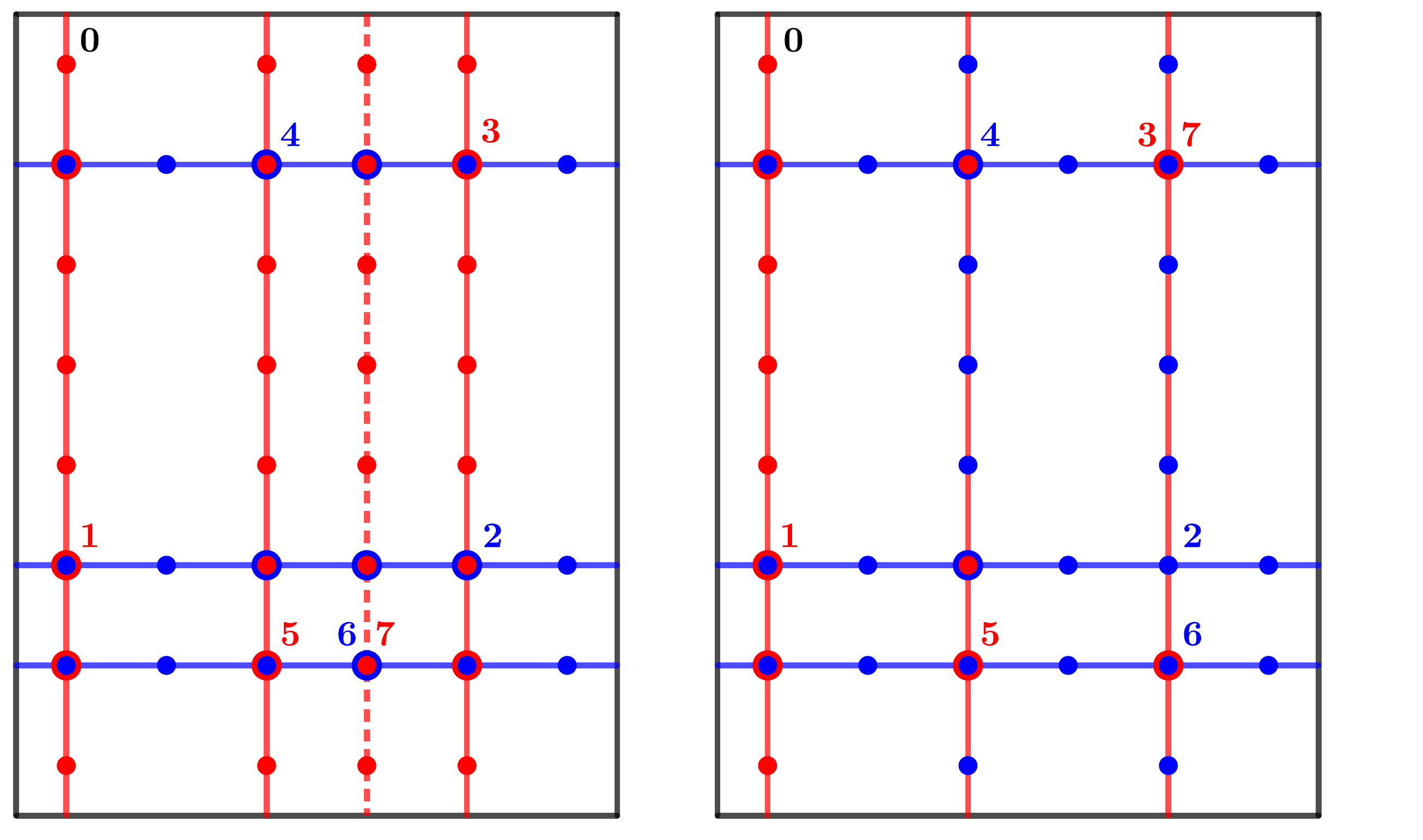}
\caption{Both  figures show an instance of the first seven steps in the \ac{BRD}.
The figure on the left describes the case in which $\stime^{\NE}_\run=6$.
To compute $\BRD_\run(7)$, the row player visits the action profiles on the red dashed lines and finds the maximum payoff at $\BRD_\run(6)$; hence $\BRD_\run(6)\in\NE_\run$. 
In this case $\resp_\run(5)$ consists of all the action profiles on the solid red and blue lines in the figure. 
Hence, $\stime^{\resp}_\run=7$ and, consequently, $\stime^{\resp}_\run-1=\stime^{\NE}_\run$.
The figure on the right describes the case in which the \ac{BRD} discovers a trap.
Since $\BRD_\run(6)$ and $\BRD_\run(3)$ are in the same column, we have  $\BRD_\run(7)=\BRD_\run(3)$. 
In this case $\stime^{\resp}_\run=6$ because $\resp_\run(4)$ consists of all the action profiles on the solid lines in the figure, except for the blue line passing through the action profiles labeled $5$ and $6$. 
Hence, $\BRD_\run(\stime^{\resp}_\run+1)=\BRD_\run(\per)$ for $\per=3\leq \stime^{\resp}_\run-3$.} 
\label{Fig:tauR}
\end{figure}

In general, the trajectory of the \ac{BRD} for all $\per\ge 0$ is completely determined by the trajectory up to the random time $\stime_{\run}^\resp$. Concerning the random time $\stime_{\run}^{\NE}$ defined in \eqref{eq:tau-NE}, notice that
\begin{enumerate}[(a)]
    \item\label{it:R-NE} in the case described in \eqref{item:NE}, we have $\stime_{\run}^{\NE}=\stime_{\run}^{\resp}-1$ if $\stime_{\run}^{\resp}>2$, and $\stime_{\run}^{\NE}=0$ if $\stime_{\run}^{\resp} =2$ (see \cref{Fig:tauR} on the left);
    \item whereas, in the case in \eqref{item:trap}, $\stime_{\run}^{\NE}=\infty$.
\end{enumerate}
By \eqref{eq:response} and \eqref{eq:def-tau-R}, we have
\begin{equation}
\label{eq:up-bound-tau-R}
\Prob(\stime_{\run}^\resp\le 2\nactionsn)=1\ ,
\end{equation}
Moreover, by \eqref{eq:up-bound-tau-R} and \eqref{it:R-NE}, if a \ac{PNE} is eventually reached, then the \ac{BRD} must visit the set $\NE_{\run}$ for the first time within $2\nactionsn-1$ steps. 
Formally,
\begin{equation}
\label{eq:general-bound-tau-NE}
\Prob(\stime_{\run}^{\NE}\le 2\nactionsn-1)=1-\Prob(\stime_{\run}^{\NE}=\infty)\ .
\end{equation}
As shown in  the following lemma, another relevant feature of \eqref{eq:response} and \eqref{eq:response-0} is that, for all $\per\ge 1$,
conditionally on the event $\{\stime_{\run}^{\resp}>\per \}$, even if the set $\resp_{\run}(\per)$ is random,  its cardinality  $\abs*{\resp_\run(\per)}$ is almost surely equal to the deterministic value $\cresp_{\run}$. 

\begin{lemma}
\label{le:size-of-R}
Fix any $\run\in\naturals$ and $\proba_\run\in[0,1]$. 
Call, for every $\per\ge 1$,
\begin{equation}
\label{eq:r-t=}
\cresp_{\run}(\per)\coloneqq 
\ceil*{\dfrac{\per+1}{2}}\nactionsAn+\floor*{\dfrac{\per+1}{2}}\nactionsBn - \floor*{\dfrac{\per+1}{2}}\ceil*{\dfrac{\per+1}{2}}.
\end{equation}
Then, for every $\per\ge 1$,
\begin{equation}
\label{eq:card-R}
\Prob\parens*{\abs{\resp_{\run}(\per)} = \cresp_{\run}(\per)\mid \stime_{\run}^{\resp}>\per}=1.
\end{equation}
\end{lemma}

\subsection{Potential games}
\label{suse:p=1}

We start studying the case of $\proba=1$, \ie the case of potential games.
A potential game cannot have traps.
Therefore, thanks to \eqref{it:R-NE}, 
\begin{equation}
\label{eq:equivalence-potential}
    \braces*{\stime_{\run}^{\resp}=\per}
    =\braces*{\stime_{\run}^{\NE}=\per -1},\quad\text{for }\per> 2, 
\end{equation}
and, together with \eqref{eq:general-bound-tau-NE} and \eqref{eq:up-bound-tau-R}, we recover the well-known fact
\begin{equation}
\label{eq:NE-always-potential}
    \Prob(\stime_{\run}^{\NE}<\infty)=1\,,\qquad \run\ge 1 .
\end{equation}
Define
\begin{align}
\label{eq:q-t-n}
\qu_{\per,\run} &\coloneqq \Prob\parens*{\stime_{\run}^{\NE} = \per \mid \stime_{\run}^{\NE} \ge \per}\quad\text{for all }\per \in \braces*{0,\dots,2 \nactionsn-1}.
\end{align}
In words, $\qu_{\per,\run}$ represents the conditional probability that a \ac{PNE} is reached at time $\per$, given that it was not reached before.
Notice that if $\per=2\nactionsn-1$, then $\qu_{\per,\run}=1$.
We start with a non-asymptotic result, which, for every $\run\in\naturals$, provides the value of $\qu_{\per,\run}$ for all $\per\le 2\nactionsn-1$.

\begin{lemma}
\label{le:bound-tau-NE}
Fix $\run\in\naturals$, let $\proba_{\run}=1$ and recal the definition of $\cresp_{\run}(\per)$ in \eqref{eq:r-t=}. 
Then, 
\begin{align}
\label{eq:q-0-n}
\qu_{0,\run} &=
\frac{1}{\nactionsAn+\nactionsBn-1},
\\
\label{eq:q-1-n}
\qu_{1,\run} &=
\frac{\nactionsAn-1}{\nactionsAn+\nactionsBn-2},
\intertext{and, for $\per \ge 2$,}
\label{eq:q-t-n=}
\qu_{\per,\run} &=
\frac{\cresp_{\run}(\per-1)}{\cresp_{\run}(\per)}.
\end{align}
\end{lemma}
As an immediate consequence of the previous lemma, we obtain the exact form of the distribution of the random time $\stime_{\run}^{\NE}$.

\begin{theorem}
\label{th:exact-distr-tau}
With the definitions of \cref{le:bound-tau-NE}, we have
\begin{equation}
\label{eq:surv-tau}
\Prob\parens*{\stime_{\run}^{\NE} > \per} =
\prod_{j=0}^{\per}(1-\qu_{j,\run}).
\end{equation}
\end{theorem}

The following proposition provides the asymptotic expectation and variance of $\stime_{\run}^{\NE}$ when the two players have the same action set.

\begin{proposition}
\label{pr:KA=KB}
If  $\proba_{\run}=1$ and $\nactionsAn=\nactionsBn$ for every $\run\in\naturals$. Then 
\begin{align}
\label{eq:lim-E-tau-NE}
\lim_{\run\to\infty} \Expect\bracks*{\stime_{\run}^{\NE}} 
&= \expo - 1 \\
\label{eq:lim-Var-tau-NE}
\lim_{\run\to\infty} \Var\bracks*{\stime_{\run}^{\NE}} 
&\approx 0.767. 
\end{align}
\end{proposition}

We point out that our expression for $\Expect\bracks*{\stime_{\run}^{\NE}}$  coincides with the one in \citet[theorem 4]{DurGau:AGT2016}, but the proof techniques exploited in \cref{th:exact-distr-tau} are completely different from the ones used by  \citet{DurGau:AGT2016}; moreover, our analysis considers the whole distribution of $\stime_{\run}^{\NE}$ and not just its expectation.

%
% Subsection ----------------------------------
%

\subsection{Games with \iid payoffs}
\label{suse:iid}

In a recent paper \citet[theorem~2.3]{AmiColHam:ORL2021} showed that, when $\proba_{\run}=0$ and $\nactionsAn=\nactionsBn$, with high probability as $\run\to\infty$, the \ac{BRD} does not converge to the set $\NE_{\run}$. We start by generalizing their result to the general setting $\nactionsAn\neq\nactionsBn$. 
\begin{theorem}
\label{th:iid-rectangular}
Let $\proba_{\run}=0$ for all $\run\in\naturals$.  Then
\begin{equation}
\label{eq:tau<infty}
\lim_{\run\to\infty}\Prob(\stime_{\run}^{\NE}<\infty)=0.
\end{equation}
\end{theorem}
Even though the result in \cref{th:iid-rectangular} can be achieved by naturally adapting the proof of \citet[theorem~2.3]{AmiColHam:ORL2021} to the rectangular case, for completeness we present a detailed proof in \cref{suse:proof-BRD-iid}. 
It is worth stressing  that, contrarily to our setting, \citet{AmiColHam:ORL2021} assume that the \ac{BRD} starts at an action profile in which the second player is already at best-response. 
The fact that we dispense with this  assumption makes the quantities appearing in our proof slightly different from those in \citet{AmiColHam:ORL2021}.

%
% Subsection ----------------------------------
%

\subsection{The general case}
\label{suse:general}
The main purpose of this paper is to complement the negative result in \cref{th:iid-rectangular} by showing that a tiny bit of (positive) correlation in the players payoffs, namely $\proba_{\run}>0$, is enough to dramatically change the picture and make it similar to the case of a potential game, \ie $\proba_\run=1$.

More precisely, the following result shows that, if $\proba_\run$ is not too small compared to $\nactionsn$, then the probability of the event $\braces*{\stime_{\run}^{\NE}=\infty}$ vanishes as $\run$ goes to infinity.
\begin{theorem}
\label{th:BRD-p-not0}
Fix a positive sequence $\proba_{\run}$. If 
\begin{equation}
\label{eq:condition-K}
		\lim_{\run\to\infty}\frac{\log(\proba_{\run})}{\log(\nactionsn )}=0,
		\end{equation}
then
		\begin{equation}\label{eq:tNE<infty}
		\lim_{\run\to\infty} \Prob\parens*{\stime_{\run}^{\NE}=\infty}=0.
		\end{equation}
\end{theorem}
Notice that the latter result is qualitative in nature: it tells us that the \ac{BRD} will eventually converge to a Nash equilibrium rather than to a trap, but does not provide any bound on the rate of convergence beyond the trivial one presented in \eqref{eq:general-bound-tau-NE}. 
The following result---which  implies \cref{th:BRD-p-not0}---provides a much better bound on the time of convergence. 
Indeed, it states that, under the condition in \eqref{eq:condition-K}, the time of convergence of the \ac{BRD} to an equilibrium can be upper bounded, with high probability, by any function diverging exponentially faster than $\proba_\run^{-1}$.
\begin{proposition}
\label{pr:pn-p-tau}
Fix a positive sequence $\proba_{\run}$ satisfying \eqref{eq:condition-K}. 
Then, for any sequence $\Per_\run$ such that 
\begin{equation}
\lim_{\run\rightarrow \infty} \Per_\run =\infty \quad\text{and}\quad 
\lim_{\run\rightarrow \infty} \frac{\log (\proba_\run)}{\log (\Per_\run)}=0,
\end{equation}
we have
\begin{equation}\label{eq:co-pn-p-tau}
		\lim_{\run\to\infty}\Prob\parens*{\stime_{\run}^{\NE}> \Per_\run} = 0\,.
	\end{equation} 
\end{proposition}

Notice that \cref{th:BRD-p-not0} follows from \cref{pr:pn-p-tau} by choosing $\Per_{\run}= 2 \nactionsn-1$.
In the particular case when $\proba_{\run}=\proba>0$ for all $\run\in\naturals$, \cref{pr:pn-p-tau} states that,  with high probability, a \ac{PNE} is reached by the \ac{BRD} in at most $\Per_{\run}$ steps, no matter how slowly the sequence $\Per_{\run}$ diverges.

\begin{remark}
Notice that, whenever $\proba_{\run}<1$, the random variable $\stime_{\run}^{\NE}$ takes value $+\infty$ with positive probability.
Therefore, even though \eqref{eq:tNE<infty} holds true, the random variable $\stime_{\run}^{\NE}$ cannot have a finite expectation.
\end{remark}

%
% Section ----------------------------------
%

\section{Proofs}
\label{se:proofs}

%
% Subsection ----------------------------------
%

\subsection*{Proofs of \cref{se:number-PNE}}

\begin{proof}[Proof of \cref{pr:number-PNE}]
By linearity of expectation, we have
\begin{equation}
\label{eq:E-W}
\Expect\bracks*{\nequi} = \nactionsA\nactionsB \Prob((1,1)\in\NE).
\end{equation}
Conditioning on the event $\braces*{\PayA(1,1)=\PayB(1,1)}$ and using the law of total probability, we get
\begin{equation}
\label{eq:P-1-1-nonasympt}
\Prob((1,1)\in\NE) = \proba \frac{1}{\nactionsA
+\nactionsB-1}+(1-\proba)\frac{1}{\nactionsA\nactionsB}.
\end{equation}
To see that \eqref{eq:P-1-1-nonasympt} holds, consider the following: 
when $\PayA(1,1)=\PayB(1,1)$, the profile $(1,1)$ is a \ac{PNE} if and only if $\PayA(1,1)$ is larger than or equal to all payoffs $\PayA(\actA,1)$ and all payoffs $\PayB(1,\actB)$  for all $\actA\in\actionsA$ and all $\actB\in\actionsB$.
By symmetry, this happens with probability $1/(\nactionsA + \nactionsB-1)$.
On the other hand, when $\PayA(1,1) \neq \PayB(1,1)$, the profile $(1,1)$ is a \ac{PNE} if and only if  $\PayA(1,1)$ is larger than or equal to all payoffs $\PayA(\actA,1)$ for all $\actA\in\actionsA$ and $\PayB(1,1)$ is larger than or equal to all payoffs $\PayB(1,\actB)$ for all $\actB\in\actionsB$.  
By independence of the payoffs and by symmetry, this happens with probability $1/\nactionsA\nactionsB$.
Plugging \eqref{eq:P-1-1-nonasympt} into \eqref{eq:E-W}, we get the result.
\end{proof}

\begin{proof}[Proof of \cref{pr:number-PNE-gen}]

By \cref{pr:number-PNE} we get
\begin{equation}
\label{eq:E=W-n-variant}
\Expect\bracks*{\nequi_{\run}} 
= \proba_{\run} \frac{\nactionsAn\nactionsBn}{\nactionsAn+\nactionsBn-1} + (1-\proba_{\run})
= (1+\smalloh(1))\proba_{\run}  \frac{\nactionsAn\nactionsBn}{\nactionsAn+\nactionsBn},
\end{equation}
where, in the asymptotic equality we used the fact that $\proba_{\run}\nactionsn\to\infty$.
This shows that $\Expect\bracks*{\nequi_{\run}}\to\infty$.

We now show that $\nequi_{\run}$ concentrates around its expectation.
To this end, we use an upper bound on the second moment  of  $\nequi_{\run}$. 
Notice that
\begin{equation}
\label{eq:second-moment-W-n}
\begin{split}
\Expect\bracks*{\nequi_{\run}^{2}} 
&= \Expect\bracks*{\parens*{\sum_{\parens{\actA,\actB}\in\actionsAn\times\actionsBn}\ind_{\parens{\actA,\actB}\in\NE_{\run}}}^{2}} \\
&= \sum_{\parens{\actA,\actB}\in\actionsAn\times\actionsBn}
\sum_{\parens{\actAalt,\actBalt}\in\actionsAn\times\actionsBn}
\Prob\parens*{\parens{\actA,\actB},\parens{\actAalt,\actBalt}\in\NE_{\run}}.
\end{split}
\end{equation}
From \eqref{eq:second-moment-W-n} it follows immediately that the computation of the second moment amounts to studying  probabilities of the form
\begin{equation*}
 \Prob\parens*{\parens{\actA,\actB},\parens{\actAalt,\actBalt}\in\NE_{\run}}\ ,\qquad \actA,\actAalt\in\actionsAn,\:\actB,\actBalt\in\actionsBn.  
\end{equation*}
We argue that there are only three relevant cases: 
\begin{enumerate}[(a)]
\item 
\label{it:a=a-b=b}    $\actA=\actAalt$ and $\actB=\actBalt$;

\item 
\label{it:a=a-bnotb}
$\actA=\actAalt$ and $\actB\neq\actBalt$ or vice versa;

\item
\label{anota-bnotb}
$\actA\neq\actAalt$ and $\actB\neq\actBalt$;
\end{enumerate}
Case~\eqref{it:a=a-b=b}  is trivial:
\begin{equation}
\label{eq:case-1}
    \Prob\parens*{\parens{\actA,\actB},\parens{\actAalt,\actBalt}\in\NE_{\run}}=\Prob\parens*{\parens{\actA,\actB}\in\NE_{\run}}\ .
\end{equation}
The continuity of $\distr$ implies that in Case~\eqref{it:a=a-bnotb} we have
\begin{equation}
\label{eq:P-ab-abprime-NE}
\Prob\parens*{\parens{\actA,\actB},\parens{\actAalt,\actBalt}\in\NE_{\run}} = 0.
\end{equation}
To analyze Case~\eqref{anota-bnotb}, we remark that, 
if $\actA\neq\actAalt$ and $\actB\neq\actBalt$, then the event $\braces*{\parens{\actAalt,\actBalt}\in\NE_{\run}}$ depends on the event $\braces*{\parens{\actA,\actB}\in\NE_{\run}}$ only through the payoffs $\PayA(\actAalt,\actB)$ and $\PayB(\actA,\actBalt)$.
Therefore
\begin{equation}
\label{eq:ab-eq-abprime-eq}
\begin{split}
\Prob\parens*{\parens{\actA,\actB},\parens{\actAalt,\actBalt}\in\NE_{\run}}
&= \Prob\parens*{\parens{\actA,\actB}\in\NE_{\run}}
\Prob\parens*{\parens{\actAalt,\actBalt}\in\NE_{\run} \mid \parens{\actA,\actB}\in\NE_{\run}} \\
&\le \Prob\parens*{\parens{\actA,\actB}\in\NE_{\run}}
\Prob\left(
\PayA(\actAalt,\actBalt) = \max_{\actAaltalt\in\left[\nactionsAn\right]\setminus\{\actA\}} \PayA(\actAaltalt,\actBalt), \right.\\
&\left.\qquad\PayB(\actAalt,\actBalt) = \max_{\actBaltalt\in\actionsBn\setminus\{\actB\}} \PayB(\actAalt,\actBaltalt) \right)\\
&= \Prob\parens*{\parens{\actA,\actB}\in\NE_{\run}}
\parens*{\proba_{\run}\frac{1}{\nactionsAn+\nactionsBn-3}
+(1-\proba_{\run})\frac{1}{(\nactionsAn-1)(\nactionsBn-1)}}\\
&=(1+\smalloh(1)) \parens*{\Prob\parens*{\parens{\actA,\actB}\in\NE_{\run}}}^{2}.
\end{split}
\end{equation}
The inequality in \eqref{eq:ab-eq-abprime-eq} stems from the fact that the event $(\actAalt,\actBalt)\in\NE_\run$ coincides with the event $\left\{\PayA(\actAalt,\actBalt) = \max_{\actAaltalt\in\left[\nactionsAn\right]} \PayA(\actAaltalt,\actBalt),\,\PayB(\actAalt,\actBalt) = \max_{\actBaltalt\in\actionsBn} \PayB(\actAalt,\actBaltalt)\right\}$, which in turn implies the event $\left\{\PayA(\actAalt,\actBalt) = \max_{\actAaltalt\in\left[\nactionsAn\right]\setminus \{\actA\}} \PayA(\actAaltalt,\actBalt),\,\PayB(\actAalt,\actBalt) = \max_{\actBaltalt\in\actionsBn\setminus \{\actB\}} \PayB(\actAalt,\actBaltalt)\right\}$. 
The independence between this latter event and $\{(\actA,\actB)\in\NE_\run\}$ proves the inequality.
The second equality in  \eqref{eq:ab-eq-abprime-eq},  follows from the fact that, when $\PayA(\actAalt,\actBalt)=\PayB(\actAalt,\actBalt)$, which happens with probability $\proba_{\run}$, the event 
\begin{equation}
\label{eq:eventUAB}
\braces*{\PayA(\actAalt,\actBalt) = \max_{\actAaltalt\in\actionsAn\setminus\{\actA\}} \PayA(\actAaltalt,\actBalt)\,, \,\PayB(\actAalt,\actBalt) = \max_{\actBaltalt\in\actionsBn\setminus\{\actB\}}\PayB(\actAalt,\actBaltalt)  }
\end{equation}
has the same probability as the event of picking the maximum among $\nactionsAn+\nactionsBn-3$ equally probable objects; 
on the other hand, when $\PayA(\actAalt,\actBalt)\neq\PayB(\actAalt,\actBalt)$, the event in \eqref{eq:eventUAB} has the same probability of picking independently the maximum of $\nactionsAn-1$ equally probably objects and the maximum of $\nactionsBn-1$ equally probably objects.
The last equality stems from \eqref{eq:P-1-1-nonasympt}.

In conclusion, plugging \cref{eq:case-1,eq:P-ab-abprime-NE,eq:ab-eq-abprime-eq} into \eqref{eq:second-moment-W-n}, we obtain
\begin{equation}
\label{eq:E-W-2-le}
\begin{split}
\Expect\bracks*{\nequi_{\run}^{2}} 
&= 
\sum_{\parens{\actA,\actB}\in\actionsAn\times\actionsBn}
\Prob\parens*{\parens{\actA,\actB}\in\NE_{\run}}  
\\
&\quad +
(1+\smalloh(1))
\sum_{\parens{\actA,\actB}\in\actionsAn\times\actionsBn}
\sum_{\substack{\parens{\actAalt,\actBalt}\in\actionsAn\times\actionsBn\\
\actA \neq \actAalt,\actB \neq \actBalt}}
 \parens*{\Prob\parens*{\parens{\actA,\actB}\in\NE_{\run}}}^{2} \\
 &<
\Expect\bracks*{\nequi_{\run}} +
(1+\smalloh(1))(\Expect\bracks*{\nequi_{\run}})^{2} \\
&= (1+\smalloh(1)) (\Expect\bracks*{\nequi_{\run}})^{2},
\end{split}    
\end{equation}
where the last equality stems from the fact that   $\Expect\bracks*{\nequi_{\run}}\to\infty$, which implies $\Expect\bracks*{\nequi_{\run}} = \smalloh((\Expect\bracks*{\nequi_{\run}})^{2})$.

By Chebyshev's inequality, we have that, for every $\varepsilon>0$, 
\begin{equation}
\label{eq:Chebyshev}
\Prob\parens*{\abs*{\nequi_{\run}-\Expect[\nequi_{\run}]} \ge \varepsilon \Expect[\nequi_{\run}]}
\le \frac{\Expect\bracks*{\nequi_{\run}^{2}}-\parens*{\Expect\bracks*{\nequi_{\run}}}^{2}}{\varepsilon^{2}\parens*{\Expect\bracks*{\nequi_{\run}}}^{2}}
=\smalloh(1),
\end{equation}
where the asymptotic upper bound follows from \eqref{eq:E-W-2-le}.
\end{proof}

%
% Subsection ----------------------------------
%

\subsection*{Proofs of \cref{suse:BRD-tau}}

\begin{proof}[Proof of \cref{le:size-of-R}]
We prove \eqref{eq:card-R} by induction. 
For $\per=1$ we have $\abs{\resp_{\run}(1)}=\nactionsAn+\nactionsBn-1$, which is trivially true. 
Assume now that \eqref{eq:card-R} holds up to $\per-1 < \stime_{\run}^{\resp}$. 
Notice that, conditioning on $\{\stime_{\run}^{\resp}>t\}$, $\BRD_{\run}(\per)$ cannot visit a row or column visited at time $1,2,\ldots,\per-1$.
Since player $\pA$  plays first, by the inductive hypothesis and thanks to the conditioning, for $0 < \per< \stime_{\run}^{\resp}$, we have almost surely, 
\begin{align}
\label{eq:r-t-1=}
\abs{\resp_{\run}(\per)}=\begin{cases}
\cresp_{\run}(\per-1)+\nactionsBn-\floor*{\dfrac{\per+1}{2}}\,, &\text{if $\per$ is odd};
\\\\
\cresp_{\run}(\per-1)+\nactionsAn-\floor*{\dfrac{\per+1}{2}}\,, &\text{if $\per$ is even}.
\end{cases}
\end{align}
Since
\begin{equation}
\label{eq:r-n-t-1}
\cresp_{\run}(\per-1) = 
\ceil*{\dfrac{\per}{2}}\nactionsAn+\floor*{\dfrac{\per}{2}}\nactionsBn - \floor*{\dfrac{\per}{2}}\ceil*{\dfrac{\per}{2}},
\end{equation}
we can rewrite the right hand side of \eqref{eq:r-t-1=} as \eqref{eq:r-t=}.
Hence, \eqref{eq:card-R} holds.  
\end{proof}

\subsection*{Proofs of \cref{suse:p=1}}
\begin{proof}[Proof of \cref{le:bound-tau-NE}]

Since $\proba_{\run}=1$, we have $\Bimatrix_{\run}^{\pA} = \Bimatrix_{\run}^{\pB}$ almost surely.
Therefore, to simplify the notation, we  write
\begin{equation}
\label{eq:U-UA-UB}
\Pay_{\run}(\actA,\actB)\coloneqq\PayA_{\run}(\actA,\actB) = \PayB_{\run}(\actA,\actB).
\end{equation}
Moreover, in a potential game
$\braces*{\stime_{\run}^{\resp}=\per}
=\braces*{\stime_{\run}^{\NE}=\per -1}$, for $\per> 2$, as in \eqref{eq:equivalence-potential}.
We start with $\per=0$.
We have
\begin{equation}
\label{eq:q-0-n=}
\begin{split}
\qu_{0,\run} 
&= \Prob\parens*{\stime_{\run}^{\NE} = 0 }\\
&= \Prob\parens*{\Pay_{\run}(1,1) = \max \braces*{\max_{\actA\in\actionsAn}\Pay_{\run}(\actA,1),\max_{\actB\in\actionsBn}\Pay_{\run}(1,\actB)}}
\\
&= \frac{1}{\nactionsAn+\nactionsBn-1}.
\end{split}
\end{equation}
Let now $\per=1$.
We have
\begin{equation}
\label{eq:P-tau-1-ge-1}
\Prob\parens*{\stime_{\run}^{\NE} \ge 1} =
\Prob\parens*{\stime_{\run}^{\NE} \ge 0} - \Prob\parens*{\stime_{\run}^{\NE} = 0} 
= 1 - \frac{1}{\nactionsAn+\nactionsBn-1}
= \frac{\nactionsAn+\nactionsBn-2}{\nactionsAn+\nactionsBn-1}.
\end{equation}
Moreover, if we define the event 
\begin{equation}
\label{eq:best-i}
\best_{\actA} \coloneqq \braces*{\text{player $\pA$'s best response to action $1$ is $\actA$}},
\end{equation}
we have
\begin{equation}
\label{eq:P-tau-1=}
\begin{split}
\Prob\parens*{\stime_{\run}^{\NE} = 1} 
&= \sum_{\actA \in \actionsAn} \Prob\parens*{\stime_{\run}^{\NE} = 1 \mid \best_{\actA}} \Prob(\best_{\actA})
\\
&= \sum_{\actA=2}^{\nactionsAn} \Prob\parens*{\Beta\parens*{\nactionsAn,1} \ge \Beta\parens*{\nactionsBn-1,1}} \frac{1}{\nactionsAn}
\\
&= \sum_{\actA=2}^{\nactionsAn} \frac{\nactionsAn}{\nactionsAn+\nactionsBn-1} \cdot \frac{1}{\nactionsAn}
\\
&= \frac{\nactionsAn-1}{\nactionsAn+\nactionsBn-1},
\end{split}
\end{equation}
where, with an abuse of notation, we have identified a random variable with its distribution.
Conditionally on  $\best_{\actA}$, we have $\stime_{\run}^{\NE} = 1$ if the payoff in $(\actA,1)$ is the largest among all payoffs in the same row. 
To get the result we have applied \cref{prop:maxunif} about the maximum of uniform independent random variables and \cref{prop:betacomp} about the probability that a $\Beta(\actA,1)$ is larger than an independent $\Beta(\actB,1)$. 
This result can be applied because the payoffs, and consequently the two Beta random variables, are independent.
Therefore, combining \eqref{eq:P-tau-1-ge-1} and \eqref{eq:P-tau-1=}, we obtain
\begin{equation}
\label{eq:q-1-n=}
\begin{split}
\qu_{1,\run} &= \Prob\parens*{\stime_{\run}^{\NE} = 1 \mid \stime_{\run}^{\NE} \ge 1}
\\
&= \frac{\Prob\parens*{\stime_{\run}^{\NE} = 1}}{\Prob\parens*{\stime_{\run}^{\NE} \ge 1}}
\\
&= \frac{\nactionsAn+\nactionsBn-1}{\nactionsAn+\nactionsBn-2} \cdot
\frac{\nactionsAn-1}{\nactionsAn+\nactionsBn-1}
\\
&=
\frac{\nactionsAn-1}{\nactionsAn+\nactionsBn-2}.
\end{split}
\end{equation}

Let now $\per \geq 2$. 
Call $\paths_{\per}$ the set of sequences $\ppath=(\ppath_{0},\ppath_{1},\ldots,\ppath_{\per})$ such that 
\begin{itemize}
\item $\ppath_{i}=(\actA_{i},\actB_{i})\in[\nactionsAn]\times [\nactionsBn]$ for $i\geq 0$,
\item $\ppath_{0}=(1,1)$,
\item if $i$ is odd, then $\actB_{i}=\actB_{i-1}$, whereas, if $i>0$ is even, then $\actA_{i}=\actA_{i-1}$,
\item there is no pair of distinct odd indices $i,j$ such that $\actB_{i}=\actB_{j}$,
\item there is no pair of distinct even indices $i,j$ such that $\actA_{i}=\actA_{j}$.
\end{itemize}
Notice that the set $\paths_{\per}$ coincides with all possible trajectories of length $\per$ of $\BRD_{\run}$ satisfying the event $\braces*{\stime_{\run}^{\NE}\geq\per}$.
Define the event
\begin{equation}
\label{eq:best-paths}
\bestpath_{\per}^{\ppath}=\braces*{\BRD_{\run}(\peralt)=\ppath_{\peralt}\text{ for } 0\leq\peralt\leq \per}.
\end{equation}
Notice that 
\begin{align}
\label{eq:timettau}
\qu_{\per,\run} \coloneqq \Prob\parens*{\stime_{\run}^{\NE}=  \per\mid \stime_{\run}^{\NE} \ge \per}=\sum_{\ppath\in\paths_{\per}} \Prob\parens*{\stime_{\run}^{\NE}= \per \mid \bestpath_{\per}^{\ppath},\,  \stime_{\run}^{\NE} \ge \per}\Prob\parens*{\bestpath_{\per}^{\ppath} \mid \stime_{\run}^{\NE} \ge \per}.
 \end{align}

The conditional probability $\Prob\parens*{\stime_{\run}^{\NE}=  \per \mid \bestpath_{\per}^{\ppath},\,  \stime_{\run}^{\NE} \ge \per}$  equals the probability that the maximum of $\cresp_\run(\per-1)$ \iid uniform random variables is bigger than the maximum of $\nactionsstar_\run-\floor*{\frac{\per+1}{2}}$ uniform random variables, where 
\begin{equation}
\label{eq:K-star}
\nactionsstar_\run=
\begin{cases}
\nactionsAn & \text{if $\per$ is even,}\\
\nactionsBn & \text{if $\per$ is odd.}
\end{cases}
\end{equation}
 
The conditioning event $\bestpath_{\per}^{\ppath}\cap\{  \stime_{\run}^{\NE} \ge \per\}$ determines the position of $\BRD_\run(\per)$ and the fact that the payoffs associated to $\cresp_\run(\per-1)$ action profiles up to time $\per$ have been computed by the \ac{BRD}. 
The payoff associated to $\BRD_\run(\per)$ is the maximum of $\cresp_\run(\per-1)$ \iid uniform random variables. 
The probability that  the action profile $\BRD_\run(\per)$ is a \ac{PNE} is the probability that its  payoff is larger than all the previously uncomputed payoffs associated to action profiles in its same column (if $\per$ is even) or in its same row (if $\per$ is odd).
This requires a comparison with $\nactionsstar_\run-\floor*{\frac{\per+1}{2}}$ uniformly distributed payoffs. 
So 
\begin{align*}
\Prob\parens*{\stime_{\run}^{\NE}=  \per\mid \bestpath_{\per}^{\ppath},\,  \stime_{\run}^{\NE} \ge \per}
&=\Prob\parens*{\Beta(\cresp_\run(\per-1),1)>\Beta\parens*{\nactionsstar_\run-\floor*{\frac{\per+1}{2}},1}}\\
&=\frac{\cresp_\run(\per-1)}{\cresp_\run(\per-1) + \nactionsstar_\run-\floor*{\frac{\per+1}{2}}}\\
&=\frac{\cresp_\run(\per-1)}{\cresp_\run(\per)}\,,
\end{align*}
where the last identity is due to \eqref{eq:r-t-1=}. 
Using \eqref{eq:timettau} we get
\begin{equation}
\label{eq:timettau2}
\begin{split}
 \Prob\parens*{\stime_{\run}^{\NE}
 =  \per\mid \stime_{\run}^{\NE} \ge \per}&=\sum_{\ppath\in\paths_{\per}} \Prob\parens*{\stime_{\run}^{\NE}=  \per\mid \bestpath_{\per}^{\ppath},\, \stime_{\run}^{\NE} \ge \per}\Prob\parens*{\bestpath_{\per}^{\ppath} \mid \stime_{\run}^{\NE} \ge \per} \\
 &=\frac{\cresp_\run(\per-1)}{\cresp_\run(\per)}\sum_{\ppath\in\paths_{\per}}\Prob\parens*{\bestpath_{\per}^{\ppath} \mid \stime_{\run}^{\NE} \ge \per}\,.
 \end{split}
 \end{equation}
 Note that $\{\stime_{\run}^{\NE} \ge \per\}=\sqcup_{\ppath\in\paths_{\per}} \bestpath_{\per}^{\ppath}$.
 Hence 
\begin{equation}
\label{eq:sum-pi-Pi}
\sum_{\ppath\in\paths_{\per}}\Prob\parens*{\bestpath_{\per}^{\ppath} \mid \stime_{\run}^{\NE} \ge \per}=\Prob\parens*{\stime_{\run}^{\NE} \ge \per\mid \stime_{\run}^{\NE} \ge \per}=1\,.
\end{equation}
Therefore, \eqref{eq:timettau2} becomes
\begin{equation*}
\label{eq:timettau3}
\Prob\parens*{\stime_{\run}^{\NE}=  \per\mid \stime_{\run}^{\NE} \ge \per}=\frac{\cresp_\run(\per-1)}{\cresp_\run(\per)}\,. \qedhere
\end{equation*}
\end{proof}

\begin{proof}[Proof of \cref{th:exact-distr-tau}]
Note that 
\begin{equation}
\label{eq:surv-tau-2}
\Prob\parens*{\stime_{\run}^{\NE} > \per}=\Prob\parens*{\stime_{\run}^{\NE} > \per\mid \stime_{\run}^{\NE} > \per-1 }\Prob\parens*{\stime_{\run}^{\NE} > \per-1}=(1-\qu_{\per,\run})\Prob\parens*{\stime_{\run}^{\NE} > \per-1}\,.
\end{equation}
Hence, by iteration, we get
\begin{equation*}
\Prob\parens*{\stime_{\run}^{\NE} > \per}= \prod_{j=1}^{\per}(1-\qu_{j,\run} )\cdot \Prob\parens*{\stime_{\run}^{\NE} >0}= \prod_{j=0}^{\per}(1-\qu_{j,\run} )\,. \qedhere
\end{equation*}
\end{proof}

\begin{proof}[Proof of \cref{pr:KA=KB}]
We first compute $ \Expect\bracks*{\stime_{\run}^{\NE}}$. 
By \cref{th:exact-distr-tau} we have
\begin{equation}
\label{eq:expect-sum}
\Expect\bracks*{\stime_{\run}^{\NE}} 
=\sum_{\per=0}^{2\nactionsn - 2} \Prob\parens*{\stime_{\run}^{\NE}>\per}=\sum_{\per=0}^{2 \nactionsn - 2}\prod_{j=0}^{\per} (1-\qu_{j,\run})\,.
\end{equation}
We split the first sum in \eqref{eq:expect-sum} into two parts:  $\per \in \braces*{0,\dots, \Per_{\run}-1 }$ and  $\per\in\braces*{\Per_{\run},\dots, 2 \nactionsn -2 }$, where $\Per_{\run}=\floor{\left(\log\nactionsn\right)^{2}}$. 
We start by showing that
\begin{equation}
\label{eq:expect-sum2}
\lim_{\run\to\infty}\sum_{\per=\Per_{\run}}^{2\nactionsn-2}\prod_{j=0}^{\per} (1-\qu_{j,\run})=0,
\end{equation}
for which it suffices to show that for $\per\geq \Per_{\run}$ 
\begin{equation}
\label{eq:expect-sum3}
\Prob(\stime_{\run}^{\NE}>\per)=\prod_{j=0}^{\per}(1-\qu_{j,\run})=o\parens*{\nactionsn^{-1}}.
\end{equation}
Notice that the sequence $\qu_{\per,\run}$, defined in \eqref{eq:q-t-n}, is increasing in $\per$.
Hence, 
for $\per \ge \Per_{\run}$, we have
\begin{equation}
\label{eq:P=small-o}
\Prob\parens*{\stime_{\run}^{\NE}>\per}\leq (1-\qu_{0,\run})(1-\qu_{1,\run})^{\Per_{\run}}
\leq \frac{1}{2^{\Per_{\run}}}=o\parens*{\nactionsn^{-1}}\,,
\end{equation}
where in the last asymptotic equality we used  $\qu_{1,\run}=1/2$ and $\Per_{\run}=\omega(\log(\nactionsn))$. 

We are left to show that
\begin{equation}\label{eq:limit-q-Tn}
\lim_{\run\to\infty}\sum_{\per=0}^{\Per_{\run}}\prod_{j=0}^{\per} (1-\qu_{j,\run})=\expo-1\,.
\end{equation}
Notice that
\begin{equation}
\label{eq:1-q}
1-\qu_{0,\run}=\frac{2 \nactionsn -2}{2\nactionsn-1}=1-\bigoh\parens*{\nactionsn^{-1}}\quad\text{and}\quad 1-\qu_{1,\run}=\frac{1}{2}\,.
\end{equation}
Moreover, since
\begin{equation}
\label{eq:r-t-r-t-1}
\cresp_\run(\per)-\cresp_\run(\per-1)=\nactionsn-\ceil*{\frac{\per}{2}}\,,
\end{equation}
we have, for all $\per\le\Per_\run$,
\begin{equation}
\label{eq:1-qt}
\begin{split}
1-\qu_{\per,\run}&=\frac{\cresp_\run(\per)-\cresp_\run(\per-1)}{\cresp_\run(\per)}\\
&=\frac{\nactionsn-\ceil*{\frac{\per}{2}}}{(\per+1)\nactionsn-\ceil*{\frac{\per+1}{2}}\floor*{\frac{\per+1}{2}}}\\
&=\frac{1}{\per+1}\cdot\frac{1-\Theta(\per/\nactionsn)}{1-\Theta(\per/\nactionsn)}\\
&=\frac{1}{\per +1}\cdot\parens*{1+\bigoh\parens*{\frac{\Per_\run}{\nactionsn}}}.
\end{split}
\end{equation}
Hence, by \eqref{eq:1-q} and \eqref{eq:1-qt}, for all $\per\le\Per_\run$,
\begin{equation}\label{eq:factT}
\begin{split}
\prod_{j=0}^{\per} (1-\qu_{j,\run})=&(1+\smalloh(1))\cdot1\cdot \frac{1}{2}\cdot \prod_{j=2}^{\per} \left[\frac{1}{j+1}\cdot \parens*{1+\bigoh\parens*{\frac{\Per_\run}{\nactionsn}}}\right]\\
=&\frac{1}{(\per+1)!} \parens*{1+\bigoh\parens*{\frac{\Per_\run^{2}}{\nactionsn}}}\\
=&(1+\smalloh(1))\frac{1}{(\per+1)!},
\end{split}
\end{equation}
where in the last two steps we used that, by definition, $\Per_\run=\smalloh(\sqrt{\nactionsn})$. 
Hence 
\begin{equation}
\parens*{1+\bigoh\parens*{\frac{\Per_\run}{\nactionsn}}}^{\per-1} \le\parens*{1+\bigoh\parens*{\frac{\Per_\run}{\nactionsn}}}^{\Per_\run}\le \parens*{1+\bigoh\parens*{\frac{\Per_\run^{2}}{\nactionsn}}}=(1+\smalloh(1)),\qquad 2\le\per\le\Per_\run.
\end{equation}
Notice that 
\eqref{eq:factT} implies
\begin{equation}
\label{eq:lim-sum-prod}
\sum_{\per=0}^{\Per_{\run}}\prod_{j=0}^{\per} (1-\qu_{j,\run})=(1+\smalloh(1))\sum_{\per=0}^{\Per_{\run}}\frac{1}{(\per+1)!}=(1+\smalloh(1))(\expo-1)\,,
\end{equation}
and \eqref{eq:limit-q-Tn} follows by taking the limit as $\run\to\infty$.
At this point  \eqref{eq:lim-E-tau-NE} follows from \eqref{eq:expect-sum2} and \eqref{eq:limit-q-Tn}.

We now prove \eqref{eq:lim-Var-tau-NE}. 
Call $\distrtau_{\run}$ the distribution function of $\stime_{\run}^{\NE}$. 
By \eqref{eq:factT}, we have, for $\per\leq\Per_\run$,
\begin{equation}
\label{eq:dist-tau}
\distrtau_{\run}(\per)=1-\Prob\parens*{\stime_{\run}^{\NE}>\per}=(1+\smalloh(1))\cdot\parens*{1-\frac{1}{(\per+1)!}}.
\end{equation}

In what follows, we will use the following lemma, whose proof can be found, for instance, in \citet[corollary~3]{OgrRus:EJOR1999}.

\begin{lemma}
\label{le:var-formula}
Let $X$ be a  random variable with finite expectation $\exptau$, finite variance, and distribution function $G$. 
Define the function $\intdist(x) \coloneqq \int_{0}^x G(\per)\diff \per$. Then 
\begin{equation}
\label{eq:var-formula}
\frac{1}{2}\Var[X]=\int_{0}^{+\infty} (\intdist(\per)-[\per-\exptau]_+)\diff \per,
\end{equation}
where $[\per-\exptau]_+=\max\{\per-\exptau,0\}$.
\end{lemma}

We now go back to the proof of \eqref{eq:lim-Var-tau-NE}. 
Since the random variable $\stime_{\run}^{\NE}$ is bounded, we can apply \cref{le:var-formula} to get
\begin{equation}
\label{eq:var-tau-n}
	\Var\bracks*{\stime_{\run}^{\NE}}
	=2\int_{0}^{\exptau} \intdist_{\run}(\per)\diff\per
	+2\int_{\exptau}^{+\infty} [\intdist_{\run}(\per)-(\per-\exptau)]\diff\per\,,
\end{equation}
where, thanks to \eqref{eq:dist-tau}, for all $\per\le \Per_{\run}$, 
\begin{equation}
\label{eq:Phi-n}
\intdist_{\run}(\per)=\int_{0}^{\per} \distrtau_{\run}(\peralt)\diff\peralt\sim \int_{0}^{\per}\left(1-\frac{1}{\floor{\peralt +1}!}\right)\diff\peralt,
\end{equation}
By explicit numerical integration, using \eqref{eq:Phi-n}
and the fact that $\exptau \coloneqq \Expect\bracks*{\stime_{\run}^{\NE}}\sim \expo-1\le \Per_\run$ for all $\run$ sufficiently large, we get
\begin{equation}
\label{eq:num-int-1}
	2\int_{0}^\exptau\intdist_{\run}(\per)\diff \per\sim 2\int_{0}^\exptau \diff \per \int_{0}^{\per}\left(1-\frac{1}{\floor{\peralt +1}!}\right)\diff \peralt \approx 0.258\,.
\end{equation}
Moreover, for all $\per\ge 0$,
\begin{equation}
\label{eq:Phi-n-approx}
\begin{split}
\intdist_{\run}(\per)-(\per-\exptau)
&=\int_{0}^{\per}(\distrtau_{\run}(\peralt)-1)\diff \peralt +\int_{0}^{+\infty}(1-\distrtau_{\run}(\peralt))\diff \peralt \\
&=\int_{\per}^{+\infty}(1-\distrtau_{\run}(\peralt))\diff \peralt
\\&=\int_{\per}^{\Per_\run\vee \per} (1-\distrtau_{\run}(\peralt))\diff \peralt+\int_{\Per_\run\vee \per}^{+\infty}(1-\distrtau_{\run}(\peralt))\diff \peralt\,.
\end{split}
\end{equation}
Since the sequence $\qu_{\per,\run}$, defined in \eqref{eq:q-t-n}, is increasing in $\per$ and $\qu_{1,\run}=1/2$,  by \eqref{eq:expect-sum3} we get, for all $\peralt\ge 0$
\begin{equation}
\label{eq:1-Hn}
1-\distrtau_{\run}(\peralt)=\Prob\parens*{\stime_{\run}^{\NE}>\peralt}\leq \frac{1}{2^\peralt}.   
\end{equation}
Hence, for all $\per\ge 0$,
\begin{equation}\label{eq:bound-H-large-s}
\int_{\Per_\run\vee\per}^{+\infty}(1-\distrtau_{\run}(\peralt))\diff \peralt\leq \frac{1}{\log 2}{2^{-(\Per_\run\vee \per)}},
\end{equation}
which goes to zero when $\run\rightarrow\infty$. 
Therefore, by \eqref{eq:dist-tau}, \eqref{eq:Phi-n-approx}, and \eqref{eq:bound-H-large-s}, we conclude that
\begin{equation}\label{eq:Phi-cases}
\intdist_{\run}(\per)-(\per-\exptau)=\begin{cases}
(1+o(1))\int_{\per}^{\Per_\run}\frac{1}{\floor{\peralt+1}!} \diff\peralt+\bigoh(2^{-\Per_\run})&\text{if }\per\le \Per_\run,\\
\bigoh(2^{-\per})&\text{if }\per>\Per_\run.
\end{cases}
\end{equation}
It is now convenient to split the second integral in \eqref{eq:var-tau-n} as follows
\begin{equation}
\label{eq:integral-split}
\begin{split}
2\int_{\exptau}^{+\infty} [\intdist_{\run}(\per)-(\per-\exptau)]\diff \per
&=2\int_{\exptau}^{\Per_\run} [\intdist_{\run}(\per)-(\per-\exptau)]\diff \per +2\int_{\Per_\run}^{+\infty} [\intdist_{\run}(\per)-(\per-\exptau)]\diff \per.
\end{split}
\end{equation}
At this point, using \eqref{eq:Phi-cases}, we can bound the second integral on the right hand side of \eqref{eq:integral-split} as follows
\begin{equation}\label{eq:Phi-cases-large}
\begin{split}
  \int_{\Per_\run}^{\infty} [\intdist_{\run}(\per)-(\per-\exptau)]\diff \per  &=  \int_{\Per_\run}^{+\infty} \bigoh(2^{-\per}) \diff\per
  = \bigoh(2^{-\Per_\run}) .
\end{split}
\end{equation}
On the other hand, the first integral on the right hand side of \eqref{eq:integral-split} can be bounded by
\begin{equation}\label{eq:Phi-cases-small}
\begin{split}
  \int_{\exptau}^{\Per_\run} [\intdist_{\run}(\per)-(\per-\exptau)]\diff \per  &=(1+o(1))\int_{\exptau}^{\Per_\run} \int_{\per}^{\Per_\run}\frac{1}{\floor{\peralt+1}!} \diff\peralt \diff \per+\bigoh(\Per_\run 2^{-\Per_\run})
  \\&=(1+o(1))\int_{\exptau}^{\Per_\run} \int_{\per}^{\Per_\run}\frac{1}{\floor{\peralt+1}!} \diff\peralt \diff \per
  \\&=(1+o(1))\int_{\exptau}^{\infty} \int_{\per}^{\infty}\frac{1}{\floor{\peralt+1}!} \diff\peralt \diff \per .
\end{split}
\end{equation}
In conclusion, by \eqref{eq:integral-split}, \eqref{eq:Phi-cases-large}, and \eqref{eq:Phi-cases-small} and numerical integration, we have
\begin{equation}
\label{eq:num-int-2}
2\int_{\exptau}^{+\infty} [\intdist_{\run}(\per)-(\per-\exptau)]\diff \per
\sim 2\int_{\exptau}^{+\infty} \diff \per\int_{\per}^{+\infty}\frac{1}{\floor{\peralt+1}!} \diff \peralt\approx 0.509\,.
\end{equation}
The combination of \eqref{eq:var-tau-n}, \eqref{eq:num-int-1}, and \eqref{eq:num-int-2} gives
\begin{equation*}
\Var\bracks*{\stime_{\run}^{\NE}}\approx 0.767\,.  \qedhere
\end{equation*}
\end{proof}

%
% Subsection ----------------------------------
%

\subsection*{Proofs of \cref{suse:iid}}\label{suse:proof-BRD-iid}

\begin{proof}[Proof of \cref{th:iid-rectangular}]
First observe that either $\stime_{\run}^{\NE} \le 2 \nactionsn -1$ or $\stime_{\run}^{\NE}=\infty$.
Call
\begin{equation}
\label{eq:BRD(t)=NE}
\brdNE_{\run}(\per) \coloneqq \braces*{\BRD_{\run}(\per)\in \NE_{\run}},\quad\text{and}\quad \badiid_{\run}(\per)\coloneqq \brdNE_{\run}(\per)^c\cap\braces*{\BRD_{\run}(\per)\not\in\resp_{\run}(\per-2) },
\end{equation}
where $\resp_{\run}(\per)$ is defined as in \eqref{eq:r-t=}.
Therefore, the statement in \eqref{eq:tau<infty} is equivalent to
\begin{equation}
\label{eq:tau<infty-equiv}
\lim_{\run\to\infty}\Prob(\brdNE_{\run}(2 \nactionsn -1))=0.
\end{equation}
Notice that
\begin{equation}
\label{eq:P-C0-C-1}
\Prob(\brdNE_{\run}(0)) = \frac{1}{\nactionsAn\nactionsBn}
\quad\text{and}\quad
\Prob(\brdNE_{\run}(1)) = \Prob(\brdNE_{\run}(0)) + \Prob(\brdNE_{\run}(1) \mid \brdNE_{\run}(0)^{c}) \Prob(\brdNE_{\run}(0)^{c}).
\end{equation}
We have
\begin{equation}
\label{eq:P-C1-C0c}
\begin{split}
\Prob(\brdNE_{\run}(1) \mid \brdNE_{\run}(0)^{c})
&= \sum_{\actA=1}^{\nactionsAn} \Prob\parens*{\brdNE_{\run}(1) \mid \brdNE_{\run}(0)^{c} \cap \BRD_{\run}(1) = (\actA,1)} 
\Prob\parens*{\BRD_{\run}(1) = (\actA,1) \mid \brdNE_{\run}(0)^{c}}\\
&= \sum_{\actA=2}^{\nactionsAn} \Prob\parens*{\brdNE_{\run}(1) \mid \brdNE_{\run}(0)^{c} \cap \BRD_{\run}(1) = (\actA,1)} 
\Prob\parens*{\BRD_{\run}(1) = (\actA,1) \mid \brdNE_{\run}(0)^{c}}\\
&= \sum_{\actA=2}^{\nactionsAn} \frac{1}{\nactionsBn}
\Prob\parens*{\BRD_{\run}(1) = (\actA,1) \mid \brdNE_{\run}(0)^{c}}\\
&= \frac{1}{\nactionsBn} \parens*{1 - \Prob\parens*{\BRD_{\run}(1) = (1,1) \mid \brdNE_{\run}(0)^{c}}}\\
&= \frac{1}{\nactionsBn} \parens*{1- \frac{\Prob\parens*{\brdNE_{\run}(0)^{c} \mid \BRD_{\run}(1) = (1,1)} \Prob\parens*{\BRD_{\run}(1) = (1,1)}}{\Prob\parens*{\brdNE_{\run}(0)^{c}}}}\\
&= \frac{1}{\nactionsBn} \parens*{1- \frac{\frac{\nactionsBn-1}{\nactionsBn}\frac{1}{\nactionsAn}}{1-\frac{1}{\nactionsAn\nactionsBn}} }= \frac{\nactionsAn-1}{\nactionsAn\nactionsBn-1}\,.
\end{split}
\end{equation}
This implies
\begin{equation}
\label{eq:P-C1}
\Prob(\brdNE_{\run}(1)) = \frac{1}{\nactionsAn\nactionsBn}
+ \frac{\nactionsAn-1}{\nactionsAn\nactionsBn-1}
\parens*{1 - \frac{1}{\nactionsAn\nactionsBn}}
=\frac{1}{\nactionsBn}\,.
\end{equation}
Notice that 
\begin{equation}
     \brdNE_{\run}(\per)\cap \brdNE_{\run}^c(\per-1)= \braces*{\BRD_{\run}(\per-1)\neq\BRD_{\run}(\per)=\BRD_{\run}(\per+1) }\subseteq \braces*{\BRD_{\run}(\per)\not\in\resp_{\run}(\per-2)}.
\end{equation}
Hence, by definition of  $\badiid_{\run}(\per-1)$, we get
\begin{equation}\label{eq:identity-CD}
    \brdNE_{\run}(\per)\cap \brdNE_{\run}^c(\per-1)= \brdNE_{\run}(\per)\cap \badiid_{\run}(\per-1).
\end{equation}
Moreover, since
\begin{equation}
\label{eq:monotone-C}
    \brdNE_{\run}(\per-1)\subseteq \brdNE_{\run}(\per),
\end{equation}
by \eqref{eq:identity-CD} and \eqref{eq:monotone-C} we conclude that, for $\per\geq 1$,
\begin{equation}
\label{eq:PCt}
\Prob(\brdNE_{\run}(\per)) = \Prob(\brdNE_{\run}(\per-1)) + \Prob(\brdNE_{\run}(\per) \mid \badiid_{\run}(\per-1)) \Prob(\badiid_{\run}(\per-1)).
\end{equation}

\begin{figure}[h]
\centering
\includegraphics[width=14cm]{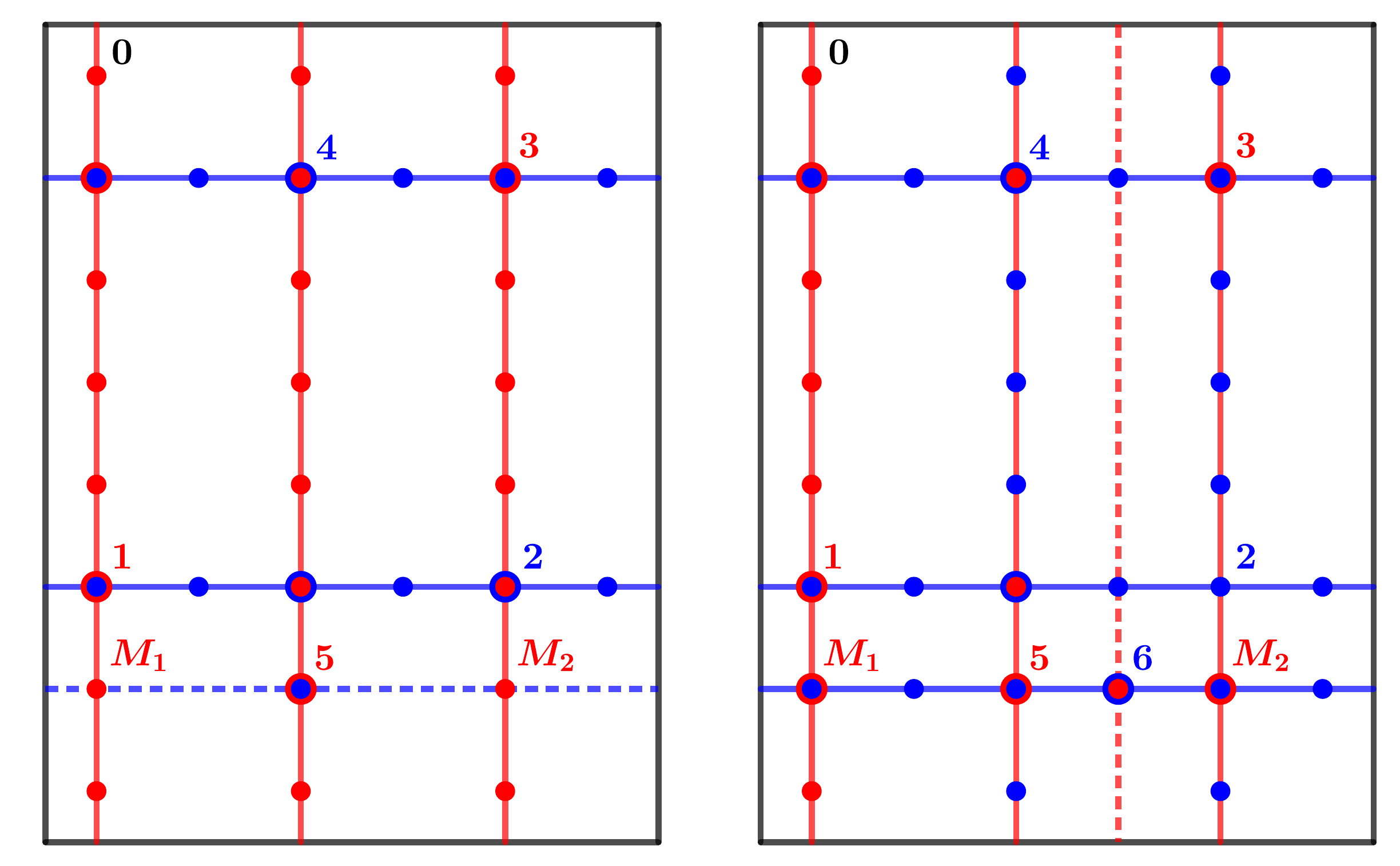}
\caption{The figure on the left shows an instance of the first five steps of the \ac{BRD}, whereas the one on the right considers an additional step. 
Given the position of $\BRD(\per)$ for $\per=1,\ldots,5$  in the figure on the left,  $\badiid_{\run}(5)$ coincides with the event that the payoff of player $\pB$ at $5$ is not the maximum of its row (the blue dashed line). 
Conditioning on $\badiid_{\run}(5)$, the probability of $\badiid_{\run}(6)$ is the product between the probability that $\BRD(6)$ is not at the action profiles $M_1$ and $M_2$, 
that is $(\nactionsBn-3)/(\nactionsBn-1)$, and, given this, the probability that the payoff of player $\pA$ at $6$ is not the maximum of its column (the red dashed line), \ie $(\nactionsAn-1)/\nactionsAn$. 
This explains \eqref{eq:P-Bt} when $\per=6$.
Similarly, conditioning on $\badiid_{\run}(5)$, $\brdNE_{\run}(6)$ is the intersection between two events: the first one is that  the position of $\BRD(6)$ does not coincide with $M_1$ and $M_2$, which has probability $(\nactionsBn-3)/(\nactionsBn-1)$; 
the second one is that the payoff of player $\pA$ at $6$ is the maximum of its column (the red dashed line), which, conditioning on the first event, has probability $1/\nactionsAn$. 
This justifies  \eqref{eq:P-Ct-Bt-1} when $\per=6$. }
\label{fig:EventoDC}
\end{figure}
By explicit computation (see \cref{fig:EventoDC}), we get
\begin{align}
\label{eq:P-Ct-Bt-1}
\Prob\parens*{\brdNE_{\run}(\per) \mid \badiid_{\run}(\per-1)}
&=
\begin{cases}
\frac{1}{\nactionsAn} \frac{\nactionsBn - \floor*{\frac{\per}{2}}}{\nactionsBn - 1} & \text{for $\per$ even,}\\
\frac{1}{\nactionsBn} \frac{\nactionsAn - \floor*{\frac{\per}{2}}}{\nactionsAn - 1} & \text{for $\per$ odd.}
\end{cases}
\end{align}
On the other hand, since
\begin{equation}
    \badiid_{\run}(\per)\subseteq \badiid_{\run}(\per-1),
\end{equation}
we have
\begin{equation}\label{eq:monotonicity-D}
\Prob\parens*{\badiid_{\run}(\per)}=\Prob\parens*{\badiid_{\run}(\per)\cap\badiid_{\run}(\per-1) }=\Prob\parens*{\badiid_{\run}(\per-1) }\Prob\parens*{\badiid_{\run}(\per)\mid\badiid_{\run}(\per-1) }.
\end{equation}
The latter conditional probability can be explicitly computed (see \cref{fig:EventoDC}), obtaining
\begin{equation}
    \label{eq:P-Bt}
\Prob\parens*{\badiid_{\run}(\per)\mid\badiid_{\run}(\per-1) }
=
\begin{cases}
\frac{\nactionsAn-1}{\nactionsAn} \frac{\nactionsBn - \floor*{\frac{\per}{2}}}{\nactionsBn - 1} & \text{for $\per$ even,}\\
\frac{\nactionsBn-1}{\nactionsBn} \frac{\nactionsAn - \floor*{\frac{\per}{2}}}{\nactionsAn - 1} & \text{for $\per$ odd.}
\end{cases}
\end{equation}
By iterating \eqref{eq:monotonicity-D} and \eqref{eq:P-Bt}, we deduce that, for all $\per \le 2 \nactionsn - 1$ odd,
\begin{equation}
\label{eq:PBt-equiv}
\begin{split}
\Prob\parens*{\badiid_{\run}(\per)}
&= \Prob\parens*{\badiid_{\run}(\per-2)} \frac{\nactionsAn - \floor*{\frac{\per}{2}}}{\nactionsAn} \frac{\nactionsBn - \floor*{\frac{\per}{2}}}{\nactionsBn}\\
&= \frac{\nactionsBn - 1}{\nactionsBn} \prod_{\substack{j=1\\j\text{ odd}}}^{\floor*{\frac{\per}{2}}}
\frac{\nactionsBn - j}{\nactionsBn} \frac{\nactionsAn - j}{\nactionsAn}\\
&\le
\prod_{\substack{j=1\\j\text{ odd}}}^{\floor*{\frac{\per}{2}}}
\parens*{1 -\frac{j}{\nactionsn}}\le
\exp\parens*{-\sum_{\substack{j=1\\j\text{ odd}}}^{\floor*{\frac{\per}{2}}} \frac{j}{\nactionsn}}\le 
\exp\parens*{-\frac{(\per-1)^{2}}{4\nactionsn}},
\end{split}
\end{equation}
where we have used the inequality $1-x\leq \expo^{-x}$.

Moreover, by \eqref{eq:P-Bt}, $\Prob\parens*{\badiid_{\run}(\per)}$ is decreasing in $\per$ and, by \eqref{eq:P-Ct-Bt-1}, $\Prob\parens*{\brdNE_{\run}(\per) \mid \badiid_{\run}(\per-1)} \le 1/\nactionsn$ for all $\per$.
Hence, iterating \eqref{eq:PCt}, we can write
\begin{equation}
\label{eq:P-Ct-equiv}
\begin{split}
\Prob\parens*{\brdNE_{\run}(\per)} 
&= \Prob\parens*{\brdNE_{\run}(\per-2)} + \Prob\parens*{\brdNE_{\run}(\per-1) \mid \badiid_{\run}(\per-2)} 
\Prob\parens*{\badiid_{\run}(\per-2)} \\
&\quad+ \Prob\parens*{\brdNE_{\run}(\per) \mid \badiid_{\run}(\per-1)} 
\Prob\parens*{\badiid_{\run}(\per-1)} \\
&\le \Prob\parens*{\brdNE_{\run}(\per-2)} + \frac{2}{\nactionsn} \Prob\parens*{\badiid_{\run}(\per-2)}.
\end{split}
\end{equation}
By applying the estimate in \eqref{eq:PBt-equiv} to $\Prob\parens*{\badiid_{\run}(\per-2)}$, we obtain, for $\per$ odd,
\begin{equation}
\label{eq:P-Ct-equiv-2}
\Prob\parens*{\brdNE_{\run}(\per)} \le \Prob\parens*{\brdNE_{\run}(\per-2)}
+ \frac{2}{\nactionsn} \exp\parens*{-\frac{(\per-3)^{2}}{4\nactionsn}},
\end{equation}
whose iteration leads to
\begin{equation}\label{eq:P-Ct-iter}
\Prob\parens*{\brdNE_{\run}(\per)} 
\le 
\Prob\parens*{\brdNE_{\run}(1)} + \frac{2}{\nactionsn} 
\sum_{\substack{j=1\\
j\text{ odd}}}^{\per-3} \exp\parens*{-\frac{j^{2}}{4\nactionsn}}
\le
\frac{1}{\nactionsn} + \frac{2}{\nactionsn} 
\sum_{\substack{j=1\\
j\text{ odd}}}^{\per-3} \exp\parens*{-\frac{j^{2}}{4\nactionsn}}\,,
\end{equation}
where in the last inequality we have used \eqref{eq:P-C1}. 

Taking $\per=2\nactionsn-1$ and using the estimate $\sum_{j=0}^{H}\exp\braces*{-j^{2}/H}=\bigoh(\sqrt{H})$ for $H\to\infty$, we conclude that
\begin{equation}
\label{eq:P-C-2KA-1}
\Prob\parens*{\brdNE_{\run}(2\nactionsn - 1)} \le \bigoh\parens*{\frac{1}{\sqrt{\nactionsn}}},
\end{equation}
which implies \eqref{eq:tau<infty}.
\end{proof}

%
% Subsection ----------------------------------
%

\subsection*{Proofs of \cref{suse:general}}
\label{suse:proof-general}

We now prove the following result, and then show that \cref{pr:pn-p-tau} immediately follows from it.
\begin{proposition}
\label{pr:BRD-p-not0}
Fix a positive sequence $\proba_{\run}$. Then, for every sequence $\Per_{\run}$ such that
\begin{equation}
\label{eq:condition-T}
\lim_{\run\to\infty}\Per_{\run} = \infty,\quad
\lim_{\run\to\infty}\frac{\log(\proba_{\run})}{\log(\Per_{\run})}=0,\quad\text{and}\quad	\lim_{\run\to\infty}\frac{\Per_{\run}}{\sqrt{\nactionsAn\wedge\nactionsBn}}=0,
	\end{equation}
	we have
\begin{equation}
\label{eq:tau<T}
	\lim_{\run\to\infty} \Prob\parens*{\stime_{\run}^{\NE}<\Per_{\run}}=1.
	\end{equation}
\end{proposition}

The proof of \cref{pr:BRD-p-not0} relies on the following lemma.
\begin{lemma}
\label{le:dom-unif}
Given a sequence $\braces*{X_{i}}_{i\in\naturals_{+}}$ of \iid random variables having a uniform distribution on $[0,1]$, consider the event 
\begin{equation}
\label{eqX-not-max}
\Xonemax=\{X_1<\max \{X_1,\ldots,X_k\}\}\,.
\end{equation}
Then for all $x\in[0,1]$ we have 
\begin{equation}
\label{eq:P-dominance}
\Prob(X_1\leq x \mid \Xonemax) \ge x = \Prob(X_1\leq x).
\end{equation}
\end{lemma}

\begin{proof}
The conditional distribution of $X_{1}$, given $\Xonemax^{c}$, is $\Beta(k,1)$, \ie
\begin{equation}
\label{eq:alf2}
\Prob(X_1\leq x\,|\,\Xonemax^c)=x^k\leq x\,.
\end{equation}
Then 
\begin{equation}
\label{eq:dom-unif00}
\begin{split}
\Prob(X_1\leq x \mid \Xonemax)&=\frac{\Prob(X_1\leq x)-\Prob(X_1\leq x\,|\,\Xonemax^c)\Prob(\Xonemax^c)}{\Prob(\Xonemax)}
\\
&=\frac{x-x^{k}\Prob(\Xonemax^c)}{1-\Prob(\Xonemax^c)}
\\
&\geq\frac{x-x\Prob(\Xonemax^c)}{1-\Prob(\Xonemax^c)}=x\,,
\end{split}
\end{equation}
where the last inequality is due to \eqref{eq:alf2}.
\end{proof}

\begin{proof}[Proof of \cref{pr:BRD-p-not0}]
Define the events
\begin{align}
\label{eq:est0}
\brdR_{\run}(\per)
&\coloneqq \braces*{\BRD_{\run}(\per) \in \resp_{\run}(\per-2)},\quad\text{for }\per \ge 2,\\ 
\label{eq:Znt}
	\brddiff_\run(\per) 
	&\coloneqq \braces*{\BRD_{\run}(\per)\neq\BRD_{\run}(\per -1)},\quad\text{for }\per \ge 1, 
\end{align}
In words, $\brdR_{\run}(\per)$ represents the event that at time $\per$ the process $\BRD_{\run}(\per)$ visits a previously visited row or column, that is,  $\brdR_{\run}(\per)=\{\stime_{\run}^{\resp}\le \per\}$,
whereas 
$\brddiff_\run(\per)$ represents the event that $\BRD_{\run}(\per-1)$ is not a \ac{PNE}.
Therefore, the sequence $\brdR_{\run}(\per)$ is increasing in $\per$, \ie
\begin{equation}
\label{eq:Et-monotone}
\brdR_{\run}(\per-1) \subseteq \brdR_{\run}(\per),
\end{equation}
whereas the sequence $\brddiff_\run(\per)$ is decreasing in $\per$, \ie
\begin{equation}
\label{eq:Zt-monotone}
\brddiff_{\run}(\per-1) \supseteq \brddiff_{\run}(\per).
\end{equation}
Then
\begin{equation}
\label{eq:PEt-less}
\begin{split}
\Prob\parens*{\brdR_{\run}(\per)\cap\brddiff_\run(\per)}&=\Prob\parens*{\brdR_{\run}(\per-1)\cap \brdR_{\run}(\per)\cap\brddiff_\run(\per)}+\Prob\parens*{\brdR_{\run}(\per-1)^c\cap \brdR_{\run}(\per)\cap\brddiff_\run(\per)}\\
&=\Prob\parens*{ \brdR_{\run}(\per-1)\cap\brddiff_\run(\per)}+\Prob\parens*{\brdR_{\run}(\per-1)^c\cap \brdR_{\run}(\per)\cap\brddiff_\run(\per)}\\
&\le\Prob\parens*{ \brdR_{\run}(\per-1)\cap\brddiff_\run(\per-1)}+\Prob\parens*{\brdR_{\run}(\per-1)^c\cap \brdR_{\run}(\per)\cap\brddiff_\run(\per)}\\
&\le \Prob\parens*{ \brdR_{\run}(\per-1)\cap\brddiff_\run(\per-1)}+\Prob\parens*{\brdR_{\run}(\per)\mid\brdR_{\run}(\per-1)^c\cap\brddiff_\run(\per)},
\end{split}
\end{equation}
where the first equality is just the law of total probabilities, the second derives from \eqref{eq:Et-monotone}, the first inequality is a consequence of \eqref{eq:Zt-monotone}, and the last stems from the definition of conditional probability.
Moreover, we claim that, 
for $\per \ge 3$, we have
\begin{equation}
\label{eq:est4}
\Prob\parens*{\brdR_{\run}(\per)\mid\brdR_{\run}(\per-1)^c\cap\brddiff_\run(\per)}\le
\begin{cases}
\frac{\floor{\frac{\per}{2}}-1}{\nactionsBn-1}&\text{if $\per$ is even,}\\
\frac{\floor{\frac{\per}{2}}-1}{\nactionsAn-1}&\text{if $\per$ is odd.}
\end{cases}
\end{equation}
The conditioning event on the \lhs of \eqref{eq:est4} represents the fact that  $\BRD_\run(\per-1)$ is neither a \ac{PNE} nor an element of $\resp_{\run}(\per-3)$.
Therefore  \eqref{eq:est4} provides a bound for the conditional probability that $\BRD_\run(\per)$ is  an element of $\resp_{\run}(\per-2)$.
To see why the bound holds, start considering  the case $\proba_{\run}=0$, where  the inequality in \eqref{eq:est4} holds as an equality. 
This is due to the fact that all payoffs are \iid\!\!. On the other extreme, when $\proba_{\run}=1$, the left hand side equals zero, since potential games do not admit traps. In the intermediate case when $\proba_{\run}\in(0,1)$, the payoffs in the row (column) of interest are not \iid\!\!.
Consider an $\pA$ payoff in a previously visited row; with probability $1-\proba_{\run}$ it is uniformly distributed on $[0,1]$ and with probability $ \proba_{\run}$ it has the law of a uniform random variable, conditioned on not being the largest payoff in its row.
A similar argument holds for $\pB$ payoffs, replacing row with column.
By \cref{le:dom-unif}, the distribution of a $\pB$ payoff on a previously visited row ($\pA$ payoff on a previously visited column)  is stochastically dominated by a uniform distribution on $[0,1]$. 
This proves the inequality in \eqref{eq:est4}.

Iterating \eqref{eq:PEt-less} we obtain
\begin{equation}
\label{eq:PEt-less-iter}
\begin{split}
\Prob\parens*{\brdR_{\run}(\per)\cap \brddiff_{\run}(\per)} 
&\le \Prob\parens*{\brdR_{\run}(2)\cap \brddiff_{\run}(2)}+  \sum_{i=3}^{\per} \frac{\floor{\frac{i}{2}}-1}{\nactionsAn \wedge \nactionsBn - 1}\\
&\le\sum_{i=1}^{\per} \frac{i}{\nactionsAn \wedge \nactionsBn - 1}
\le \frac{\per^{2}}{\nactionsAn \wedge \nactionsBn - 1},
\end{split}
\end{equation}
where for the last bound we used the fact that $\Prob\parens*{\brdR_{\run}(2)\cap \brddiff_{\run}(2)} =0$ (since $\{\brdR_\run(2)\}=\{\brddiff_{\run}(2)^c\}$) and the fact $\sum_{i=1}^\per i=\per(\per+1)/2$.
Call $\stime_{\run}^{\cyclet}$ the stopping time
\begin{equation}
\label{eq:tau-trap}
\stime_{\run}^{\cyclet} \coloneqq \inf\braces*{\per \ge 4 \colon 
\brdR_{\run}(\per) \cap \brddiff_\run(\per) \text{ holds}
},
\end{equation}
that is, $\stime_{\run}^{\cyclet}$ is the first time that the \ac{BRD} re-visits an element of a trap. 
Hence, \eqref{eq:PEt-less-iter} says that for all $\proba_{\run}\in[0,1]$, if $\Per_{\run} = \smalloh(\sqrt{\nactionsAn \wedge \nactionsBn})$, then, as $\run\to\infty$, 
\begin{equation}
\label{eq:tau-cycle-P}
\Prob(\stime_{\run}^{\cyclet} \le \Per_{\run}) \to 0.
\end{equation}
By definition of $\stime_{\run}^{\cyclet}$, the stopping time $\stime_{\run}^\resp$ defined in \eqref{eq:def-tau-R} can be rewritten as
\begin{equation}
\label{eq:tau}
\stime_{\run}^\resp \coloneqq 
\begin{cases}
    \stime_{\run}^{\cyclet} \wedge (\stime_{\run}^{\NE}+1)\,, &\text{if }(1,1)\not\in \NE_\run\,,
    \\2\,, &\text{if }(1,1)\in \NE_\run\,.
\end{cases}
\end{equation}    

Notice that, for every sequence $\Per_{\run} = \smalloh\parens*{\sqrt{\nactionsAn \wedge \nactionsBn}}$, if we show that
\begin{equation}
\label{eq:P-tau-T}
\Prob\parens*{\stime_{\run}^{\resp} \le \Per_{\run}} \to 1, 
\end{equation}
then from \eqref{eq:tau-cycle-P} and \eqref{eq:tau} it follows that
\begin{equation}
	\label{eq:P-tau-T-2}
 \Prob\parens*{\stime_{\run}^{\NE} \le \Per_{\run}} \to 1.
\end{equation}
Let 
\begin{equation}
\label{eq:equal-payoffs}
\eqpay_{\run} \coloneqq \braces*{(\actA,\actB) \in \actionsAn \times \actionsBn \colon \PayA_{\run}(\actA,\actB)=\PayB_{\run}(\actA,\actB)},
\end{equation}
be the set of action profiles that give the same payoff to the two players.
Fix now a sequence $\Per_{\run} = \smalloh\parens*{\sqrt{\nactionsAn \wedge \nactionsBn}}$ and define, for every integer $\peralt\in\braces*{1,\dots,\Per_{\run}-1}$,
\begin{equation}
		\label{eq:F-T-s}
\FTs_{\run}^{\peralt,\Per_{\run}} 
\coloneqq \braces*{\exists \per \le \Per_{\run}-\peralt\ \text{ s.t. }\BRD_{\run}(\per+\peralt') \in \eqpay_{\run},\quad \forall\peralt'\in\braces{0,\dots,\peralt}}.
	\end{equation}
The event $\FTs_{\run}^{\peralt,\Per_\run}$ occurs if there exists an interval of $\peralt$ consecutive steps before $\Per_{\run}$ in which the \ac{BRD} visits only elements in $\eqpay_{\run}$.

For every sequence $(\peralt_\run)_{\run\in\naturals}$ such that $\peralt_\run\le\Per_{\run}$ for every $\run$, we have
\begin{equation}
\label{eq:P-tau>T}
\Prob(\stime_{\run}^{\resp} > \Per_{\run}) = 
\Prob\parens*{\braces*{\stime_{\run}^{\resp} > \Per_{\run}} \cap \FTs_{\run}^{\peralt_\run,\Per_{\run}}} +
\Prob\parens*{ \braces*{\stime_{\run}^{\resp} > \Per_{\run}} \cap\braces*{\FTs_{\run}^{\peralt_\run,\Per_{\run}}}^{c}}.
\end{equation}
First we show that the first term on the \rhs of \eqref{eq:P-tau>T} goes to zero as $\peralt_\run\to \infty$. 
To this end, it is enough to show that, under the event $\FTs_{\run}^{\peralt_\run,\Per_{\run}}$, there exists some $\per \le \Per_{\run}-\peralt_\run$ such that  the best-responding player's payoff at time $\per+\peralt_\run$ is stochastically larger than a $\Beta(u,1)$ random variable, with
\begin{equation}
\label{eq:def-u}
	u\coloneqq \cresp_{\run}(\peralt_\run)-\frac{\Per_\run^{2}}{4}.
\end{equation}	
Notice that, for $\run\in\naturals$ large enough,
\begin{equation}
\label{eq:u-ineq}
u \geq \frac{\peralt_\run}{2}(\nactionsAn+\nactionsBn)-\frac{\Per_\run^{2}}{2}\ge \frac{\peralt_\run}{3}(\nactionsAn+\nactionsBn),
\end{equation}
where in the first inequality we used \eqref{eq:r-t=} and the fact that $\peralt_\run\le \Per_\run$; the second inequality holds for all sufficiently large $\run$ by our choice of $\Per_\run$.  
Indeed, after the interval of $\peralt_\run$ consecutive steps in which the \ac{BRD} visits only elements of $\eqpay_{\run}$, the \ac{BRD} visits an action profile $(\actA,\actB)$ such that
\begin{itemize}
\item $(\actA,\actB)\in \eqpay_{\run} $, that is, both players receive the same payoff;
	
\item this common payoff is the largest of a set of random variables, among which at least $u$ are \iid$\Unif([0,1])$ (see  \cref{Fig:Fs} for more details);
\end{itemize}
hence, the stochastic domination follows.

\begin{figure}[h]
\centering
\includegraphics[width=14cm]{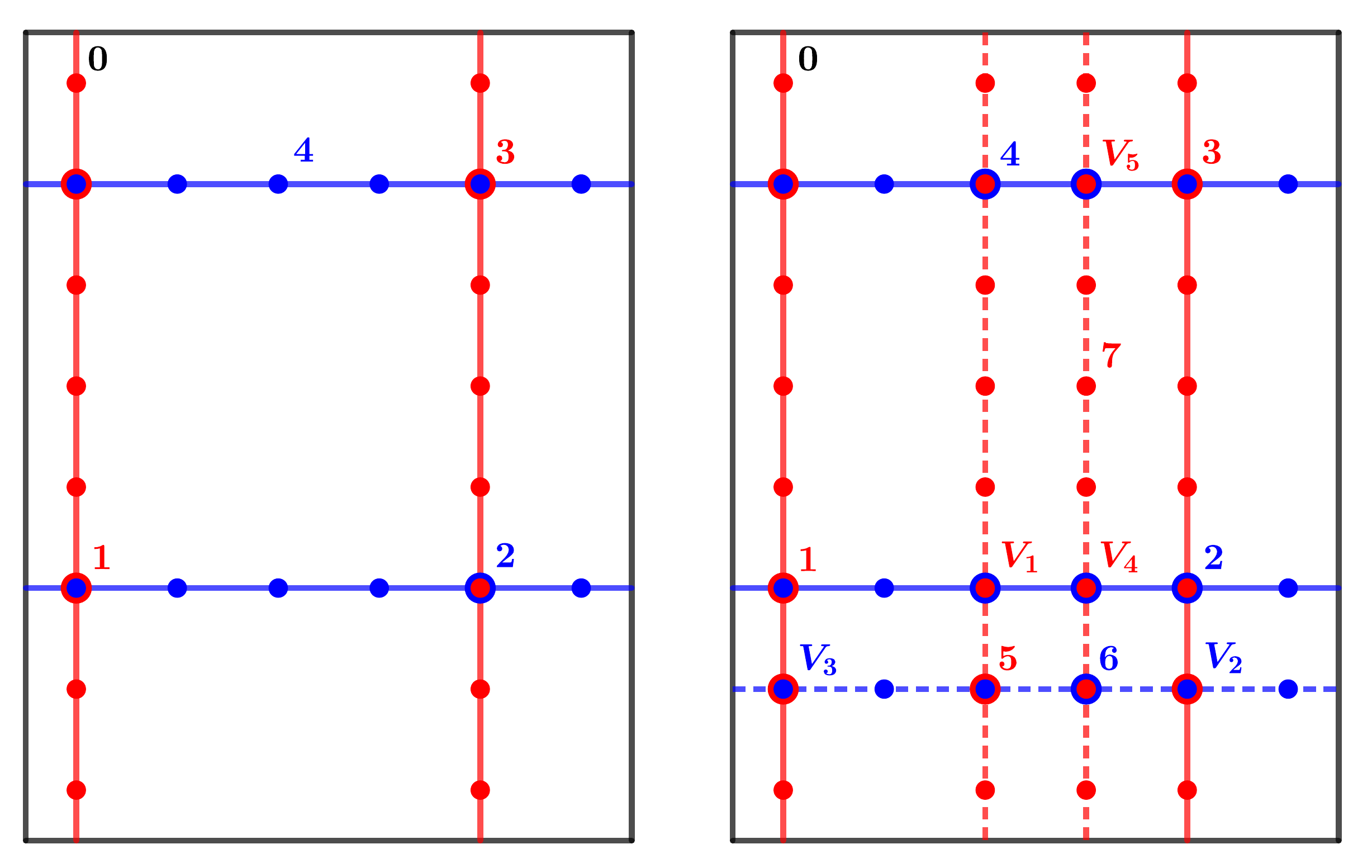}
\caption{The left figure shows an instance of the first four steps of the \ac{BRD}. 
The right figure shows how the dynamics proceeds after time $4$.
The numbered action profiles lying on the dashed lines, \ie $5$, $6$, and $7$, give the same payoff to the row and column player.
Note that in this case the event $\FTs_{\run}^{\peralt_\run,\Per_{\run}}$ occurs with $\Per_{\run}=7$ and $\peralt_\run=3$. 
Indeed, there exists $\per\leq \Per_\run-\peralt_\run$ (in this case $\per=4$) such that the \ac{BRD} visits only action profiles in $\eqpay_{\run}$ for $\peralt_\run$ consecutive steps. 
As a consequence, the payoff of the row player at the action profile $7$ is the maximum of the payoffs of the row player in the action profiles lying on the red dashed lines and of the payoffs of the column player in the action profiles lying on the blue dashed lines.  
Such payoffs are all \iid $\Unif([0,1])$ except for the payoffs associated to the action profiles $4,V_1,\ldots,V_5$, for which we have additional information.  
Since the number of such exceptional action profiles is at most $\left(\Per_\run/2\right)^{2}$ and the total number of action profiles lying on the dashed lines is $\cresp_\run(\peralt_\run)$, we have that the payoff of the row player at the action profile $7$ is the maximum of at least $u$ \iid $\Unif([0,1])$ random variables, where $u$ is defined as in \eqref{eq:def-u}.}
\label{Fig:Fs}
\end{figure}

Therefore, by \cref{prop:maxunif}, the probability that $\BRD_{\run}(\per+\peralt_\run)=\BRD_{\run}(\per+\peralt_\run+1)$, that is, $\BRD_{\run}(\per+\peralt_\run)$ is a \ac{PNE}, is bounded from below by the probability that a $\Beta(u,1)$ random variable is larger than the maximum of $\nactionsAn\vee \nactionsBn-1$ \iid random variables with a uniform distribution on $[0,1]$.
Hence, by \cref{prop:betacomp}, we get
\begin{equation}
\label{eq:est-B}
\Prob\parens*{\braces*{\stime_{\run}^{\resp} > \Per_{\run}} \cap \FTs_{\run}^{\peralt_\run,\Per_{\run}}} \le 
1- \frac{\left\lfloor\frac{\peralt_\run}{3}\right\rfloor (\nactionsAn+\nactionsBn)}{\left\lfloor\frac{\peralt_\run}{3}\right\rfloor (\nactionsAn+\nactionsBn) + (\nactionsAn\vee \nactionsBn-1)}\le 
\frac{3}{\peralt_\run+3},
\end{equation} 
which goes to zero  as $\peralt_\run\to\infty$. 
We now show that, under the assumption in \eqref{eq:condition-T},  it is possible to find a sequence $(\peralt_\run)_\run$  such that
\begin{equation}
\label{eq:various-cond-T-n}	\lim_{n\to\infty}	\peralt_\run=\infty,\qquad	\lim_{n\to\infty} \frac{\peralt_\run}{\Per_{\run}}= 0,
\end{equation}
and
\begin{equation}
\label{eq:P-tau>T-second-part}
\lim_{n\to\infty}	\Prob\parens*{ \braces*{\stime_{\run}^{\resp} > \Per_{\run}} \cap\braces*{\FTs_{\run}^{\peralt_\run,\Per_{\run}}}^{c}}\to 0.
\end{equation}
Notice that under the event $\braces*{\stime_{\run}^{\resp}>\Per_{\run}}$, the probability of the event $\braces*{\FTs_{\run}^{\peralt_\run,\Per_{\run}}}^{c}$ can be upper bounded by the probability that, splitting the interval $\braces*{0,\dots, \Per_{\run}}$ into subintervals of length $\peralt_\run$, none of them is such that the \ac{BRD} visits only $\eqpay_{\run}$ in that subinterval. 
Therefore,
\begin{equation}
\label{eq:P-tau-T-intervals}
\Prob\parens*{\braces*{\stime_{\run}^{\resp} > \Per_{\run}} \cap \braces*{\FTs_{\run}^{\peralt_\run,\Per_{\run}}}^{c}} 
\le \Prob\parens*{\Binomial\parens*{\floor*{\frac{\Per_{\run}}{\peralt_\run}},\proba_{\run}^{\peralt_\run}}=0}
= (1-\proba_{\run}^{\peralt_\run})^{\floor*{\Per_{\run}/\peralt_\run}}.
\end{equation}
If $\Per_{\run} = \smalloh(\sqrt{\nactionsAn \wedge \nactionsBn})$, then the term on the \rhs of \eqref{eq:P-tau-T-intervals} goes to zero whenever
\begin{equation}
\label{eq:lim-T}
\lim_{\run\to\infty}\frac{\Per_{\run}}{\peralt_\run}\proba_{\run}^{\peralt_\run}= \infty.
\end{equation}
A necessary and sufficient condition for \eqref{eq:lim-T} is
\begin{equation}
\label{eq:lim-T-log}
\lim_{\run\to\infty}\log(\Per_{\run})-\log(\peralt_\run)-\peralt_\run \log(\proba_{\run}^{-1})= \infty,
\end{equation}
or,  equivalently,
\begin{equation}
\label{eq:lim-T-log2}
\lim_{\run\to\infty}\log(\Per_{\run})\left[1-\frac{\log(\peralt_\run)}{\log(\Per_{\run})}-\frac{\peralt_\run \log(\proba_{\run}^{-1})}{\log(\Per_{\run})}\right]= \infty.
\end{equation}
Thanks to \eqref{eq:condition-T}, we can choose, \eg
\begin{equation}\label{eq:s-n-instance}
\peralt_\run= \log(\Per_{\run})\wedge \log\left(\frac{\log(\Per_\run)}{\log(\proba_{\run}^{-1})}\right)\to\infty.
\end{equation}
Since, for every diverging positive  sequence $(a_n)_{n\ge 0}$, it holds that $\log(a_n)/a_n\to 0$, by the definition in \eqref{eq:s-n-instance} we deduce that
\begin{equation}
\label{eq:review1}
\frac{\log(\peralt_{\run})}{\log(\Per_{\run})} \le \frac{\log(\log \Per_{\run})}{\log(\Per_{\run})} \to 0 \quad\text{as}\quad \run \to \infty \,,
\end{equation}
and
\begin{equation}
\label{eq:review2}
\frac{\peralt_{\run}\log (\proba_\run^{-1})}{\log(\Per_{\run})} 
\le
\frac{\log\left(\dfrac{\log(\Per_{\run})}{\log(\proba_\run^{-1})}  \right)}{\dfrac{\log(\Per_{\run})}{\log(\proba_\run^{-1})}} \to 0 \quad\text{as}\quad \run\to \infty\,.
\end{equation}
Coupling \eqref{eq:review1} and \eqref{eq:review2}, we immediately validate \eqref{eq:lim-T-log2}.
\end{proof}
    
\begin{proof}[Proof of \cref{pr:pn-p-tau}]
Assume that \eqref{eq:condition-K} is satisfied. 
Choose any sequence $\Per_\run$ such that 
\begin{equation}
\label{eq:hp}
\lim_{\run\to\infty}\frac{\log(\proba_{\run})}{\log(\Per_{\run})}=0.
\end{equation}
There are two cases:
\begin{itemize}
\item 
If
\begin{equation*}
\lim_{\run\to\infty} \frac{\Per_{\run}} {\sqrt{\nactionsAn\wedge\nactionsBn}} = 0,          
\end{equation*} 
 then \eqref{eq:co-pn-p-tau} follows by \cref{pr:BRD-p-not0}.
        
\item If instead
\begin{equation}
\label{eq:case2}
\limsup_{\run\to\infty}\frac{\Per_{\run}}{\sqrt{\nactionsAn\wedge\nactionsBn}}>0,
\end{equation}
then we can define the 
$\Per_{\run}'=\Per_\run\wedge \ceil{\nactionsAn\wedge\nactionsBn}^{1/3}$.
Notice that
\begin{itemize}
\item 
$\lim_{\run\to\infty}\Per_\run'=\infty$;

\item 
by the fact that $\log(\Per_\run')\ge\frac{1}{3}\log(\nactionsAn\wedge\nactionsBn)\wedge\log(\Per_\run)$, combined with  \eqref{eq:condition-K} and \eqref{eq:hp}, we deduce that 
\begin{equation*}
\lim_{\run\to\infty}\frac{\log(\proba_{\run})}{\log(\Per_\run')}=0;
\end{equation*}
            
\item 
moreover, by the definition of $\Per_\run'$ we have 
\begin{equation*}
\lim_{\run\to\infty}\frac{\Per_\run'}{\sqrt{\nactionsAn\wedge\nactionsBn}}=0. 
\end{equation*}
\end{itemize}
        
Hence, thanks to \cref{pr:BRD-p-not0} we get
\begin{equation}
\lim_{\run\to\infty} \Prob\parens*{\stime_{\run}^{\NE}<\Per_\run'}=1.
\end{equation}
To conclude the proof, it is enough to see that for every $\run\in\naturals$
\begin{equation*}
\Prob\parens*{\stime_{\run}^{\NE}<\Per_{\run}}\ge \Prob\parens*{\stime_{\run}^{\NE}<\Per_\run'}.
\qedhere
\end{equation*}
\end{itemize}
\end{proof}

%
% Section ----------------------------------
%

\section{Conclusions and open problems}
\label{se:conclusions}

We have considered a model of two-person games with random payoffs that parametrically interpolates  potential games and games with \iid payoffs. 
The interpolation acts locally on each payoff profile. 
We have studied both the asymptotic behavior of the random number of \aclp{PNE} of the game and the asymptotic behavior of \acl{BRD}, as the number of actions for each player diverges.
The type of model that we chose requires combinatorial tools for its analysis. 

We see this paper as a first attempt to provide a parametric model for random games where the payoffs are not independent, but have some structure that depends on a locally acting parameter.
Several extensions and variations of this model are conceivable and will be the object of our future research.
For instance:

\begin{enumerate}[(i)]
\item
It would be interesting to have a clearer view of the phase transition taking place at $\proba=0$. 
In particular, it would be important to investigate the existence of a sequence $\proba_\run \to 0$  such that the probability that a \ac{BRD} does not lead to a \ac{PNE} converges to a value smaller than $1$.   

\item 
Games with more than two players could be studied.

\item
With more than two players, different types of deviator rules in \ac{BRD} could be considered, \eg round-robin, random order, etc..

\item
The behavior of better-response dynamics could be studied and compared to \acl{BRD}, along the lines of \citet{AmiColHam:ORL2021}.

\item
When we deal with the number of \aclp{PNE}, we studied a form of Law of Large Numbers. 
The existence of a Central Limit Theorem could be explored.
\end{enumerate}

%
% Section ----------------------------------
%

\appendix

\gdef\thesection{\Alph{section}} % corrected redefinition of "\thesection"
\makeatletter
\renewcommand\@seccntformat[1]{\appendixname\ \csname the#1\endcsname.\hspace{0.5em}}
\makeatother

\section{List of symbols}
\label{se:symbols}

\begin{longtable}{p{.13\textwidth} p{.82\textwidth}}

$\best_{\actA}$ & $\braces*{\text{player $\pA$'s best response to action $1$ is $\actA$}}$, defined in \eqref{eq:best-i}\\
$\BRD$ & \acl{BRD}\\
$\brdNE_{\run}(\per)$ & $\braces*{\BRD_{\run}(\per)\in \NE_{\run}}$, defined in \eqref{eq:BRD(t)=NE}\\
$\badiid_{\run}(\per)$ & $\braces*{\BRD_{\run}(\per)\not\in \NE_{\run},\ \BRD_{\run}(\per)\not\in\resp_{\run}(\per-2) }$, defined in \eqref{eq:BRD(t)=NE}\\
$\brdR_{\run}(\per)$ & $\braces*{\BRD(\per)\in \resp_{\run}(\per-2)}$, defined in \eqref{eq:est0}\\ 
$\distr$ & uniform  distribution function  on $[0,1]$\\

$\bestpath_{\per}^{\ppath}$ & $\braces*{\BRD_{\run}(\peralt)=\ppath_{\peralt}\text{ for } 0\leq\peralt\leq \per}$, defined in \eqref{eq:best-paths}\\
$\distrtau_{\run}$ & distribution function of $\stime_{\run}^{\NE}$\\
$\FTs_{\run}^{\peralt,\Per_{\run}}$ &
$\braces*{\exists \per \le \Per_{\run}-\peralt\ \text{ s.t. }\BRD_{\run}(\per+\peralt') \in \eqpay_{\run}, \forall\peralt'\in\braces{0,\dots,\peralt}}$, defined in \eqref{eq:F-T-s}\\

$\nactionsAn$ & number of player $\pA$'s actions in the game $\Bimatrix_{\run}$\\
$\nactionsBn$ & number of player $\pB$'s actions in the game $\Bimatrix_{\run}$\\

$\nactionsn$ & $\min(\nactionsAn,\nactionsBn)$, defined in \eqref{eq:K-n}\\
$\actionsA$ & action set of player $\pA$\\
$\actionsB$  & action set of player $\pB$\\
$\Xonemax$ & $\{X_1<\max \{X_1,\ldots,X_k\}\}$, defined in \eqref{eqX-not-max}\\
$\Trap$ & trap, defined in \cref{de:trap}\\
$\NE$ & set of \aclp{PNE}\\

$\proba_{\run}$ & probability that $\PayA(\actA,\actB) = \PayB(\actA,\actB)$ in the game $\Bimatrix_{\run}$\\

$\qu_{\per,\run}$ & $\Prob\parens*{\stime_{\run}^{\NE} = \per \mid \stime_{\run}^{\NE} \ge \per}$, defined in \eqref{eq:q-t-n}\\

$\cresp_{\run}(\per)$ & $\ceil*{\dfrac{\per+1}{2}}\nactionsAn+\floor*{\dfrac{\per+1}{2}}\nactionsBn - \floor*{\dfrac{\per+1}{2}}\ceil*{\dfrac{\per+1}{2}}$, defined in \eqref{eq:r-t=}\\
$\resp_{\run}(\per)$ & defined in \eqref{eq:response}\\
$\eqpay_{\run}$ & $\braces*{(\actA,\actB) \in \actionsAn \times \actionsBn \colon \PayA_{\run}(\actA,\actB)=\PayB_{\run}(\actA,\actB)}$, defined in \eqref{eq:equal-payoffs}\\

$\per$ & (discrete) time\\

$\PayA$ & player $\pA$'s payoff function\\
$\PayB$ & player $\pB$'s payoff function\\

$\nequi_{\run}$ & number of \aclp{PNE} in $\Bimatrix_{\run}$\\
$\brddiff_\run(\per)$ & $\braces*{\BRD_{\run}(\per)\neq\BRD_{\run}(\per -1)}$, defined in \eqref{eq:Znt}\\

$\paths_{\per}$ & set of possible paths for $\BRD_{\run}$ up to time $\per$\\

$\stime_{\run}^{\NE}$ & first time the \ac{BRD} visits a \ac{PNE}, defined in \eqref{eq:tau-NE}\\ 
$\stime_{\run}^{\cyclet}$ & first time the \ac{BRD} re-visits an element of a trap, defined in \eqref{eq:tau-trap}\\ 
$\stime^{\resp}_\run$ & $\min\braces*{\per\ge 2:\ \BRD_\run(\per)\in\resp_\run(\per-2)}$, defined in \eqref{eq:def-tau-R}\\
$\intdist_{\run}(\per)$ & $\int_{0}^{\per} \distrtau_{\run}(\peralt)\diff\peralt$, defined in \eqref{eq:Phi-n}\\
$\potential$ & potential function, defined in \eqref{eq:potential}\\
\end{longtable}

\section{Beta distribution}
\label{se:beta}
We report two well-known results about Beta distributions. 
For the sake of completeness, we add their simple proofs.
\begin{proposition}
\label{prop:maxunif}
Let $X_1,\ldots,X_k$ be \iid random variables having a uniform distribution on $[0,1]$ and let $M_k \coloneqq \max_{i\in\braces{1,\dots,k}}X_i$. 
Then $M_k$ has distribution $\Beta(k,1)$.
\end{proposition}

\begin{proof}
For any $t\in[0,1]$ we have
\begin{equation*}
\Prob(M_k\leq t)=[\Prob(X_1\leq t)]^k=t^k,
\end{equation*}
\ie a $\Beta(k,1)$ distribution function. 
\end{proof}

\begin{proposition}
\label{prop:betacomp}
Let $X$ and $Y$ be independent random variables with distributions $\Beta(a,1)$ and $\Beta(b,1)$, respectively.
Then 
\begin{equation}
\label{eq:PX>Y}
\Prob(X>Y)=\frac{a}{a+b}\,.
\end{equation}
\end{proposition}

\begin{proof}
We have
\begin{align*}
\Prob(X>Y)&=\int_0^1 \parens*{\int_0^t a t^{a-1}\cdot b s^{b-1}\diff s} \diff t =\int_0^1 at^{a-1}\cdot t^b\diff t=\frac{a}{a+b}\,.
\qedhere
\end{align*}
\end{proof}

\section*{Acknowledgments}
\label{se:acknowledments}
The authors  deeply thank three reviewers for their extremely careful reading of the manuscript and their insightful suggestions.
Hlafo Alfie Mimun and Marco Scarsini are members of GNAMPA-INdAM.

\subsection*{Funding}
 
Hlafo Alfie Mimun and Marco Scarsini's research was supported by the GNAMPA project CUP\_E53C22001930001 ``Limiting behavior of stochastic dynamics in the Schelling segregation model'' and the Italian MIUR PRIN project 2022EKNE5K   ``Learning in Markets and Society.''
Matteo Quattropani thanks the German Research Foundation (project number 444084038, priority program SPP2265) for financial support.

\bibliographystyle{apalike}
\bibliography{bibperturbed}

\end{document}